\documentclass{amsart}

\usepackage{amsmath,amssymb,amscd,amsthm,latexsym,accents,stmaryrd}
\usepackage{amsfonts}
\usepackage[latin1]{inputenc}

\usepackage{graphicx}
\usepackage{float}
\usepackage{subfig}
\usepackage{caption}

\usepackage[all]{xy}
\usepackage{enumerate}
\usepackage{amsaddr}
\usepackage{tikz}
\usetikzlibrary{decorations.text}
\usetikzlibrary{patterns}
\usepackage{mathtools}
\usetikzlibrary{automata,positioning}

\newtheorem{theor}{Theorem}[subsection]
\newtheorem{lem}[theor]{Lemma}
\newtheorem{prop}[theor]{Proposition}

\theoremstyle{definition}

\newtheorem{ex}[theor]{Example}

\newcommand{\we}{\stackrel{\sim}{\rightarrow}}

\def\R{{\mathbb{R}}}

\def\im{{\rm{im}\,}}

\def\length{{\rm{length}}}
\def\A{{\mathcal{A}}}
\def\B{{\mathcal{B}}}
\def\C{{\mathcal{C}}}

\def\deg{{\rm{deg}}}

\def\star{{\rm{star}}}

\begin{document}

\title{Topological abstraction of higher-dimensional automata}

\author{Thomas Kahl}

\address{Centro de Matem\'atica,
Universidade do Minho, Campus de Gualtar,\\
4710-057 Braga,
Portugal. 
}

\email{kahl@math.uminho.pt
}



\begin{abstract}
Higher-dimensional automata constitute a very  expressive model for concurrent systems. In this paper, we discuss ``topological abstraction" of higher-dimensional automata, i.e., the replacement of HDAs by smaller ones that can be considered equivalent from both a computer scientific and a topological point of view. By definition, topological abstraction  preserves the homotopy type, the trace category, and the homology graph of an HDA. We establish conditions under which cube collapses yield topological abstractions of HDAs.
\end{abstract}

\maketitle

\section{Introduction}

\subsection{Higher-dimensional automata}

A higher-dimensional automaton (HDA) is an automaton with nicely incorporated squares and higher-dimensional cubes, which represent independence of actions. More formally, an HDA  over a monoid $M$ is a precubical set with initial and final states and with 1-cubes labeled by elements of $M$ such that opposite edges of 2-cubes have the same label. Higher-dimensional automata provide a powerful model for concurrent systems in which many classical formalisms to describe such systems can be interpreted \cite{CridligGoubault, GaucherProcess, vanGlabbeek, GoubaultDomains,  GoubaultJensen,  GoubaultMimram}. The concept of higher-dimensional automaton goes back to Pratt \cite{Pratt}. The notion defined above is essentially the one introduced by van Glabbeek (see \cite{vanGlabbeek}).

Let us indicate by means of an example how HDAs can be used to model concurrent systems whose component processes are given by program graphs or transition diagrams in the sense of \cite{BaierKatoen, MannaPnueli}. The construction is based on the concurrent interpretation \cite{vanGlabbeek} of the transition system semantics for such systems  \cite{BaierKatoen, MannaPnueli}. Consider the program graph on two boolean variables $x$ and $y$ depicted in Figure \ref{PG}. 
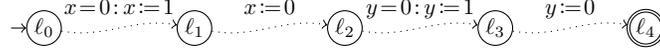
\begin{figure}[t]
\begin{tikzpicture}[initial text={},on grid] 
     
 \node[state,initial by arrow, initial where=left, initial distance=0.2cm,minimum size=0pt,inner sep =1pt,fill=white] (q_0)   {\scalebox{0.85}{$\ell_0$}}; 
 
 \node[state,minimum size=0pt,inner sep =1pt,fill=white] (q_1) [right=of q_0,xshift=1cm] {\scalebox{0.85}{$\ell_1$}};
    
 \node[state,minimum size=0pt,inner sep =1pt,fill=white] (q_2) [right=of q_1,xshift=1cm] {\scalebox{0.85}{$\ell_2$}};
   
 \node[state,minimum size=0pt,inner sep =1pt,fill=white] (q_3) [right=of q_2,xshift=1cm] {\scalebox{0.85}{$\ell_3$}};
   
 \node[state,accepting,minimum size=0pt,inner sep =1pt,fill=white] (q_4) [right=of q_3,xshift=1cm] {\scalebox{0.85}{$\ell_4$}};

    \path[->] 
    (q_0) edge[above,dotted,inner sep =5pt,out=350,in=170] node {\scalebox{0.85}{$x\!=\!0\colon\!x\!\coloneqq \!1$}} (q_1)
    (q_1) edge[above,dotted,inner sep =5pt,out=350,in=170] node {\scalebox{0.85}{$x\!\coloneqq \!0$}} (q_2)
    (q_2) edge[above,dotted,out=350,in=170] node {\scalebox{0.85}{$y\!=\!0\colon\!y\!\coloneqq \!1$}} (q_3)
    (q_3) edge[above,dotted,out=350,in=170] node {\scalebox{0.85}{$y\!\coloneqq \!0$}} (q_4);

\end{tikzpicture} 
\caption{A program graph on two boolean variables $x,y$. In the first and third instructions, the assignment action is only executable when the guard condition indicated before the colon holds.}\label{PG}
\end{figure}
The nodes of the graph symbolize  locations in a program, and the arrows represent  instructions, which are guarded or unguarded assignments of values to the variables. Now consider the concurrent system consisting of two processes executing the program given by this program graph, and suppose that the shared variables $x$ and $y$ are both initially zero. An HDA representing the accessible (reachable) part of this system is depicted in Figure \ref{ex1}. 
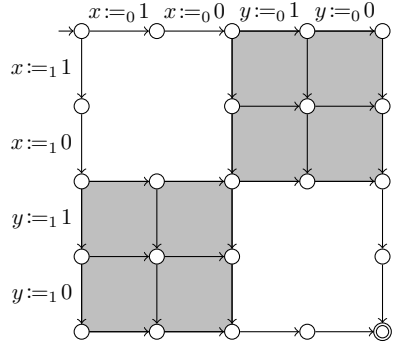
\begin{figure}[] 
\begin{tikzpicture}[initial text={},on grid]

\path[draw, fill=lightgray] (0,0)--(2,0)--(2,2)--(0,2)--cycle
(2,2)--(4,2)--(4,4)--(2,4)--cycle;

 \node[state,minimum size=0pt,inner sep =2pt,fill=white,label={[label distance = -0.1cm]330:\scalebox{0.85}{}}] (q_0)   {}; 
    
   \node[state,minimum size=0pt,inner sep =2pt,fill=white,label={[label distance = -0.1cm]330:\scalebox{0.85}{}}] (q_1) [right=of q_0,xshift=0cm] {};

\node[state,minimum size=0pt,inner sep =2pt,fill=white,label={[label distance = -0.1cm]330:\scalebox{0.85}{}}] (q_2) [right=of q_1,xshift=0cm] {};   
   
   \node[state,minimum size=0pt,inner sep =2pt,fill=white,label={[label distance = -0.1cm]330:\scalebox{0.85}{}}] [above=of q_0, yshift=0cm] (q_3)   {};
   
   \node[state,minimum size=0pt,inner sep =2pt,fill=white,label={[label distance = -0.1cm]330:\scalebox{0.85}{}}] (q_4) [right=of q_3,xshift=0cm] {};
   
   \node[state,minimum size=0pt,inner sep =2pt,fill=white,label={[label distance = -0.1cm]330:\scalebox{0.85}{}}] (q_5) [right=of q_4,xshift=0cm] {}; 
   
      \node[state,minimum size=0pt,inner sep =2pt,fill=white,label={[label distance = -0.1cm]330:\scalebox{0.85}{}}] (q_6) [above=of q_3,yshift=0cm] {};
      
      \node[state,minimum size=0pt,inner sep =2pt,fill=white,label={[label distance = -0.1cm]330:\scalebox{0.85}{}}] (q_7) [above=of q_4,yshift=0cm] {};
      
      \node[state,minimum size=0pt,inner sep =2pt,fill=white,label={[label distance = -0.1cm]330:\scalebox{0.85}{}}] (q_8) [above=of q_5,yshift=0cm] {};
      
      \node[state,minimum size=0pt,inner sep =2pt,fill=white,label={[label distance = -0.1cm]330:\scalebox{0.85}{}}] (q_9) [above=of q_6,yshift=0cm] {};
      
      \node[state,minimum size=0pt,inner sep =2pt,fill=white,label={[label distance = -0.1cm]330:\scalebox{0.85}{}}] (q_10) [above=of q_8,yshift=0cm] {};
      
      \node[state,minimum size=0pt,inner sep =2pt,initial,initial distance=0.2cm,fill=white,label={[label distance = -0.1cm]330:\scalebox{0.85}{}}] (q_11) [above=of q_9,yshift=0cm] {};
      
      \node[state,minimum size=0pt,inner sep =2pt,fill=white,label={[label distance = -0.1cm]330:\scalebox{0.85}{}}] (q_12) [right=of q_11,xshift=0cm] {};
      
      \node[state,minimum size=0pt,inner sep =2pt,fill=white,label={[label distance = -0.1cm]330:\scalebox{0.85}{}}] (q_13) [above=of q_10,yshift=0cm] {};

 \node[state,minimum size=0pt,inner sep =2pt,fill=white,label={[label distance = -0.1cm]330:\scalebox{0.85}{}}] (p_0) [right=of q_2,xshift=0cm]  {}; 
    
   \node[state,minimum size=0pt,inner sep =2pt,accepting,fill=white,label={[label distance = -0.1cm]330:\scalebox{0.85}{}}] (p_1) [right=of p_0,xshift=0cm] {};

   \node[state,minimum size=0pt,inner sep =2pt,fill=white,label={[label distance = -0.1cm]330:\scalebox{0.85}{}}] (p_4) [above=of p_1,xshift=0cm] {};

      \node[state,minimum size=0pt,inner sep =2pt,fill=white,label={[label distance = -0.1cm]330:\scalebox{0.85}{}}] (p_6) [above=of p_0,yshift=1cm] {};
      
      \node[state,minimum size=0pt,inner sep =2pt,fill=white,label={[label distance = -0.1cm]330:\scalebox{0.85}{}}] (p_7) [above=of p_4,yshift=0cm] {};
      
      \node[state,minimum size=0pt,inner sep =2pt,fill=white,label={[label distance = -0.1cm]330:\scalebox{0.85}{}}] [above=of p_7, yshift=0cm] (p_8)   {};
      
      \node[state,minimum size=0pt,inner sep =2pt,fill=white,label={[label distance = -0.1cm]330:\scalebox{0.85}{}}] (p_9) [above=of p_6,yshift=0cm] {};

      \node[state,minimum size=0pt,inner sep =2pt,fill=white,label={[label distance = -0.1cm]330:\scalebox{0.85}{}}] (p_11) [above=of p_9,yshift=0cm] {};
      
      \node[state,minimum size=0pt,inner sep =2pt,fill=white,label={[label distance = -0.1cm]330:\scalebox{0.85}{}}] (p_12) [right=of p_11,xshift=0cm] {};

    \path[->] 
    (q_0) edge[above]  node {} (q_1)
    (q_1) edge[above]  node {} (q_2)
	(q_2) edge[above]  node {} (p_0)    
    (p_0) edge[above]  node {} (p_1)
    
	(q_3) edge[left]  node {\scalebox{0.85}{$y\!\coloneqq_1 \! 0$}} (q_0)
    (q_5) edge[right]  node {} (q_2)
    (q_4) edge[left]  node {} (q_1)    
    (p_4) edge[left]  node {} (p_1)
    
    (q_3) edge[above]  node {} (q_4)
    (q_4) edge[above]  node {} (q_5)
    
    (q_6) edge[left]  node {\scalebox{0.85}{$y\!\coloneqq_1 \! 1$}} (q_3)
    (q_7) edge[right]  node {} (q_4)
    (q_8) edge[left]  node {} (q_5)    
    (p_7) edge[left]  node {} (p_4)
    
    (q_6) edge[above]  node {} (q_7)
    (q_7) edge[above]  node {} (q_8)
    (q_8) edge[above]  node {} (p_6)
    (p_6) edge[above]  node {} (p_7)
    
    (q_9) edge[left]  node {\scalebox{0.85}{$x\!\coloneqq_1 \! 0$}} (q_6)
    (q_10) edge[right]  node {} (q_8)
    (p_9) edge[left]  node {} (p_6)    
    (p_8) edge[left]  node {} (p_7)
    
    (q_10) edge[above]  node {} (p_9)
    (p_9) edge[above]  node {} (p_8)
    
    (q_11) edge[left]  node {\scalebox{0.85}{$x\!\coloneqq_1 \! 1$}} (q_9)
    (q_13) edge[right]  node {} (q_10)
    (p_11) edge[left]  node {} (p_9)    
    (p_12) edge[left]  node {} (p_8)
    
    (q_11) edge[above]  node {\scalebox{0.85}{$x\!\coloneqq_0 \! 1$}} (q_12)
    (q_12) edge[above]  node {\scalebox{0.85}{$x\!\coloneqq_0 \! 0$}} (q_13)
    (q_13) edge[above]  node {\scalebox{0.85}{$y\!\coloneqq_0 \! 1$}} (p_11)
    (p_11) edge[above]  node {\scalebox{0.85}{$y\!\coloneqq_0 \! 0$}} (p_12)
    ;

\end{tikzpicture}  
\caption{HDA model of the accessible part of the system consisting of two copies of the process given by the program graph in Figure \ref{PG}. Parallel arrows are supposed to have the same label.}\label{ex1}
\end{figure}
The vertices of the HDA correspond to the accessible global states of the system, which are quadruples whose components are  local states of the processes and values of the variables. The HDA has an initial and a final state, which correspond to the global states  $(\ell_0,\ell_0,0,0)$ and  $(\ell_4,\ell_4,0,0)$ respectively. The directed edges starting in a given vertex represent the actions of the processes that are enabled in the corresponding global state. These actions are indicated in the transition labels, indexed by the respective process IDs. The squares are introduced in order to indicate independence of actions. The concept of independence represented by the squares is closely related to the ones used in the context of partial order reduction (see e.g. \cite{BaierKatoen, ClarkeGrumbergPeled, Godefroid, Peled}): two actions are independent in a state where both are enabled if they can be executed one after the other in either order and the two sequential executions of the actions lead to the same state. In general, an HDA may also contain higher-dimensional cubes in order to indicate the independence of more than two actions.

\subsection{Topological abstraction}
Concurrent systems can be modeled mathematically at different abstraction levels. A small and abstract model of a  system is, of course, easier to understand and analyze than a big one that comes with a lot of implementation details and which, because of the state explosion problem, might not even fit into the memory of any available  computer. On the other hand, a model of a high abstraction level cannot be expected to capture all properties of the system. It should, however, contain the essential information on the system. In this paper, we adopt the point of view that the essential information on a system is its  overall topological structure and discuss what one may call topological abstraction, namely the replacement of a model by a more abstract but topologically equivalent one. To make this more precise, let us briefly describe the three aspects of the topological structure of an HDA that we will require to be preserved under topological abstraction.

1. Consider the system represented by the HDA in Figure \ref{ex1}. According to the independence square in the upper right corner, the last action of process $0$ and the first action of process $1$ are independent in the starting state of that square, and it is irrelevant for the system which of the actions comes first in a sequential execution of both. In this sense, the two adjacent execution paths from the initial to the final state passing through the endpoints of the upper right square can be considered equivalent. The equivalence relation generated by adjacency is closely related to trace equivalence, as defined in Mazurkiewicz trace theory \cite{Mazurkiewicz, Mazurkiewicz2, WinskelNielsen}, and it can also be considered a  concept of homotopy for directed, non-reversible paths, which is why it is called \emph{directed homotopy} or  \emph{dihomotopy} \cite{FajstrupGR, Goubault}. The dihomotopy classes of execution paths between important states of an HDA, such as initial and final states, define the \emph{trace category} of the HDA \cite{weakmor}. The trace category is the first ingredient of the topological structure of an HDA that we will demand to be invariant under topological abstraction. Topological abstraction is thus compatible with the principle of partial order reduction to replace models of concurrent systems by  smaller ones with fewer representatives of the equivalence classes of execution paths. We remark that the trace category is essentially the same as Bubenik's fundamental bipartite graph  \cite{BubenikExtremal}. The trace category of the HDA in Figure \ref{ex1} has four morphisms, i.e., dihomotopy classes of paths, from the initial to the final state, which can be represented by the four maximal paths going through the middle point. 

2. An important topological feature of the HDA in Figure \ref{ex1} is that it has two holes representing the fact that the shared variables $x$ and $y$ are accessed in mutual exclusion. This kind of topological information on an HDA is contained in its homology.  We will demand that topological abstraction preserves the homotopy type of an HDA. This  ensures in particular that a topological abstraction of an HDA has the same homology as the original one.

3. Since the 1-skeleton of an HDA, i.e., the underlying transition system, is a directed graph, the homology of an HDA comes equipped with a supplementary structure. In the HDA in Figure \ref{ex1}, for example, the hole in the upper left corner ``comes before'' the hole in the lower right corner. This reflects the fact that both processes access the variables in the same order. We use the terminology that the homology class corresponding to the upper left hole \emph{points} to the homology class corresponding to the lower right hole. This pointing relation turns the homology of an HDA into a directed graph, which is called the \emph{homology graph} of the HDA \cite{hgraph}. The homology graph is the third piece of the topological structure of an HDA that is required to be invariant under topological abstraction. The homology graph permits one to distinguish, for instance, the HDA in Figure \ref{ex1} and the homotopy equivalent HDA representing the accessible part of the concurrent system consisting of one process given by the program graph in Figure \ref{PG} and one process given by the same program graph with the variables $x$ and $y$ switched.

Figure \ref{abs} shows two topological abstractions of the HDA in Figure \ref{ex1}. 
\begin{figure} 
\subfloat[]
{ 

\begin{tikzpicture}[initial text={},on grid]

\path[draw, fill=lightgray] (0,0)--(2,0)--(2,2)--(0,2)--cycle
(2,2)--(4,2)--(4,4)--(2,4)--cycle;

 \node[state,minimum size=0pt,inner sep =2pt,fill=white] (q_0)   {};

\node[state,minimum size=0pt,inner sep =2pt,fill=white] (q_2) [right=of q_0,xshift=1cm] {};

      \node[state,minimum size=0pt,inner sep =2pt,fill=white] (q_6) [above=of q_0,yshift=1cm] {};

      \node[state,minimum size=0pt,inner sep =2pt,fill=white] (q_8) [above=of q_2,yshift=1cm] {};

      \node[state,minimum size=0pt,inner sep =2pt,initial,initial distance=0.2cm,fill=white] (q_11) [above=of q_6,yshift=1cm] {};

      \node[state,minimum size=0pt,inner sep =2pt,fill=white] (q_13) [above=of q_8,yshift=1cm] {};

   \node[state,minimum size=0pt,inner sep =2pt,accepting,fill=white] (p_1) [right=of q_2,xshift=1cm] {};

      \node[state,minimum size=0pt,inner sep =2pt,fill=white] (p_7) [above=of p_1,yshift=1cm] {};

      \node[state,minimum size=0pt,inner sep =2pt,fill=white] (p_12) [right=of q_13,xshift=1cm] {};

    \path[->] 
    (q_0) edge[above]  node {} (q_2)
    (q_2) edge[above]  node {} (p_1)
    
	(q_6) edge[left]  node[align=left]{\scalebox{0.85}{$y\!\coloneqq_1 \! 1;$}\\\scalebox{0.85}{$y\!\coloneqq_1 \! 0$}} (q_0)
    (q_8) edge[right]  node {} (q_2)
    (p_7) edge[left]  node {} (p_1)
    
    (q_6) edge[above]  node {} (q_8)
    (q_8) edge[above]  node {} (p_7)
    
    (q_11) edge[left]  node[align=left]{\scalebox{0.85}{$x\!\coloneqq_1 \! 1;$}\\\scalebox{0.85}{$x\!\coloneqq_1 \! 0$}} (q_6)
	(q_13) edge[right]  node {} (q_8)
    (p_12) edge[left]  node {} (p_7)
   
    (q_11) edge[above]  node[align=left]{\scalebox{0.85}{$x\!\coloneqq_0 \! 1;x\!\coloneqq_0 \! 0$}} (q_13)
    (q_13) edge[above]  node[align=left]{\scalebox{0.85}{$y\!\coloneqq_0 \! 1;y\!\coloneqq_0 \! 0$}} (p_12)
    ;
\end{tikzpicture}
\label{abs1}
}
\subfloat[]
{
\begin{tikzpicture}[initial text={},on grid]

    \node[state,minimum size=0pt,accepting,inner sep =2pt,fill=white] (p_12) [,xshift=-0.5cm] {};

      \node[state,minimum size=0pt,inner sep =2pt,fill=white] (q_8) [above=of p_12,left=of p_12,xshift=-0.5cm,yshift=2cm] {};

      \node[state,minimum size=0pt,inner sep =2pt,initial,initial distance=0.2cm,fill=white] (q_0)   [above=of q_8,left=of q_8,xshift=-0.5cm,yshift=2cm] {};

    \path[->] 
    (q_0) edge[bend left,right]   node[align=center] {\scalebox{0.85}{$x\!\coloneqq_0 \! 1;x\!\coloneqq_0 \! 0;$}\\\scalebox{0.85}{$\quad x\!\coloneqq_1 \! 1;x\!\coloneqq_1 \! 0$}} (q_8)

	(q_0) edge[bend right, left] node[align=center] {\scalebox{0.85}{$x\!\coloneqq_1 \! 1;x\!\coloneqq_1 \! 0;$}\\\scalebox{0.85}{$\quad x\!\coloneqq_0 \! 1;x\!\coloneqq_0 \! 0$}}  (q_8)
    
    (q_8) edge[bend left,right]  node[align=center] {\scalebox{0.85}{$y\!\coloneqq_0 \! 1;y\!\coloneqq_0 \! 0;$}\\\scalebox{0.85}{$\quad y\!\coloneqq_1 \! 1;y\!\coloneqq_1 \! 0$}} (p_12)
    
	(q_8) edge[bend right,left]  node[align=center] {\scalebox{0.85}{$y\!\coloneqq_1 \! 1;y\!\coloneqq_1 \! 0;$}\\\scalebox{0.85}{$\quad y\!\coloneqq_0 \! 1;y\!\coloneqq_0 \! 0$}} (p_12)
   ;
\end{tikzpicture}
\label{abs2}
}
\caption{Topological abstractions of the HDA in Figure \ref{ex1}} \label{abs}
\end{figure}
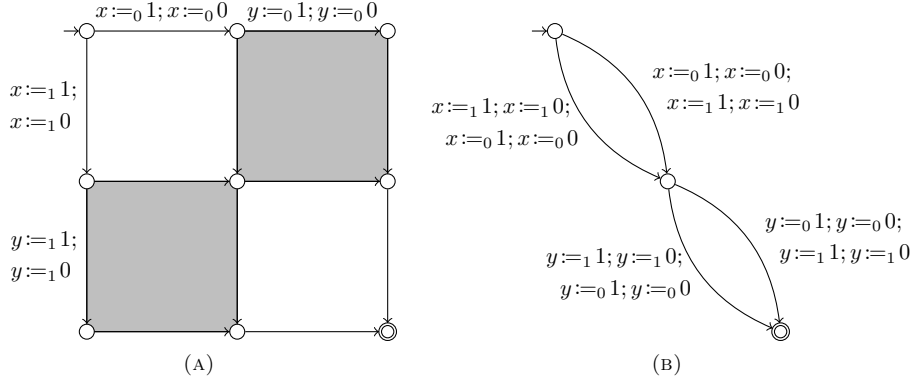 
The three HDAs are topologically equivalent in the sense described above:  they have the same trace category, the same homotopy type, and the same homology graph. The reader should note  that other concepts of topological equivalence for HDAs and related models of concurrency have been defined in the literature (see e.g. \cite{BubenikW, GaucherFlowModel, GaucherGoubault, reldi, Krishnan,  RaussenInvariants}). The HDA in Figure \ref{abs2} is more abstract than the one in Figure \ref{abs1}, and it is actually the most abstract HDA topologically equivalent to the HDA in Figure \ref{ex1}. The concept of topological abstraction gives rise to a preorder relation for HDAs, which is quite different from classical preorder relations such as simulation. Some related  preorders have been studied in \cite{weakmor}. One of them is \emph{homeomorphic abstraction}, which is related to $\mbox{T-}$homotopy equivalence in the sense of Gaucher and Goubault \cite{GaucherT, GaucherGoubault} and which is essentially always stronger than topological abstraction. For example, the HDA in Figure \ref{abs1} is a homeomorphic abstraction of the HDA in Figure \ref{ex1}.

\subsection{Cube collapses}
An important question is, of course, how to compute a  topological abstraction of an HDA. In this paper, we establish conditions under which it is possible to collapse a cube in an HDA in order to obtain a topological abstraction of it. We consider two kinds of cube collapses, namely elementary collapses, where one removes a cube and a free face of it, and vertex star collapses, where the star of a vertex of a cube is eliminated.  Combined with cube merging operations for the computation of  homeomorphic abstraction, cube collapses can be used to develop algorithms to compute topological abstractions of HDAs. In the last section of this paper, we will illustrate this possibility informally by means of the example of Peterson's mutual exclusion algorithm. A detailed description of a formal algorithm will be given elsewhere.

A first attempt to use cube collapsing operations in order to reduce higher-dimensional automata to topologically equivalent smaller ones is described in \cite{dicubes2d}. In that paper, the at least intuitively stronger concept of directed homotopy equivalence is considered instead of topological abstraction. The main shortcoming of that approach  is that it is restricted to two-dimensional HDAs and cannot be extended to higher dimensions. The approach of the present  paper, based on the concept of topological abstraction, is not bound to these limitations.

\subsection{Outline of the paper}

The main results of the paper are contained in Sections \ref{TC}, \ref{HG}, and \ref{topabs}. Section \ref{secHDA} contains preliminaries about HDAs and precubical sets. Section \ref{wr} is devoted to weak regularity, which is a convenient technical condition on precubical sets. In the final section, we illustrate our results with the example of Peterson's mutual exclusion algorithm. 

The subject of Section \ref{TC} is the trace category. We establish conditions under which the trace category is invariant under cube collapses. We also discuss dihomotopy invariant properties of HDAs and make precise how they are determined by the trace category. In Section \ref{HG}, we investigate conditions under which cube collapses preserve the homology graph. In Section \ref{topabs}, we define the preorder relation of topological abstraction and relate it to homeomorphic abstraction. Based on our results on the trace category and the homology graph, we give conditions under which cube collapses yield topological abstractions of HDAs.  

\section{Precubical sets and HDAs} \label{secHDA}

This section briefly presents some fundamental concepts and facts about precubical sets and higher-dimensional automata.

\subsection{Precubical sets} \label{precubs}

A \emph{precubical set} is a graded set $P = (P_n)_{n \geq 0}$ with  \emph{boundary} or \emph{face operators} $d^k_i\colon P_n \to P_{n-1}$ $(n>0,\;k= 0,1,\; i = 1, \dots, n)$ satisfying the relations $d^k_i\circ d^l_{j}= d^l_{j-1}\circ d^k_i$ $(k,l = 0,1,\; i<j)$ \cite{FahrenbergThesis, FajstrupGR, GaucherGoubault, Goubault, weakmor, hgraph}. If $x\in P_n$, we say that $x$ is of  \emph{degree} or \emph{dimension} $n$ and write $\deg(x) = n$. The elements of degree $n$ are called the \emph{$n$-cubes} of $P$. The elements of degree $0$ are also called the \emph{vertices} or the \emph{nodes} of $P$. We say that a face $d^k_ix$ of a cube $x$ is  a \emph{front (back) face} of $x$ if $k=0$ $(k=1)$. A morphism of precubical sets is a morphism of graded sets that is compatible with the boundary operators. 

Let $k$ and $l$ be two integers such that $k \leq l$. The \emph{precubical interval}  $\llbracket k,l \rrbracket$ is the  precubical set defined by $\llbracket k,l \rrbracket_0 = \{k,\dots , l\}$, $\llbracket k,l \rrbracket_1 =  \{[k,k+ 1], \dots , [l- 1,l]\}$, $d_1^0[j-1,j] = j-1$, $d_1^1[j-1,j] = j$ and $\llbracket k,l \rrbracket_{\geq 2} = \emptyset$. We shall use the abbreviations $\llbracket k,l \llbracket = \llbracket k,l\rrbracket \setminus \{l\}$ and $\rrbracket k,l \rrbracket = \llbracket k,l\rrbracket \setminus \{k\}$. Infinite precubical intervals such as $\llbracket 0, +\infty \llbracket$ are defined analogously.

The category of precubical sets is a monoidal category. The \emph{precubical $n$-cube} is the $n$-fold tensor product $\llbracket 0,1\rrbracket^ {\otimes n}$. The only element of degree $n$ in $\llbracket 0,1\rrbracket^ {\otimes n}$ will be denoted by $\iota_n$. Given an element $x$ of degree $n$ of a precubical set $P$, the unique morphism of precubical sets $\llbracket 0,1\rrbracket ^{\otimes n}\to P$ such that $\iota_n \mapsto x$ is denoted by $x_{\sharp}$. We say that $x$ is \emph{regular} (or \emph{non-self-linked} \cite{FajstrupGR}) if $x_{\sharp}$ is injective.

A \emph{path of length $k$} in a precubical set $P$ is a morphism of precubical sets $\omega \colon \llbracket 0,k \rrbracket \to P$. The set of paths in $P$ is denoted by $P^{\mathbb I}$. If $\omega \in P^{\mathbb I}$ is a path of length $k$, we write $\length(\omega) = k$. If it is defined, the \emph{concatenation} of two paths $\omega$ and $\nu$ is denoted by $\omega \cdot \nu$. Note that for any path $\omega \in P^{\mathbb I}$ of length $k \geq 1$ there exists a unique sequence $(x_1, \dots , x_k)$ of elements of $P_1$ such that $d_1^0x_{j+1} = d_1^1x_j$ for all $1\leq j < k$ and $\omega = x_{1\sharp} \cdots x_{k\sharp}$. If there exists a path in $P$ from a vertex $v$ to a vertex $w$, we say that $w$ is \emph{accessible} from $v$ and write $v \to_P w$. Occasionally, we also consider \emph{infinite paths} in $P$, which are morphisms of precubical sets $\llbracket 0, +\infty \llbracket \to P$. 

Two paths $\omega$ and $\nu $ in a precubical set $P$ are said to be \emph{elementarily dihomotopic} (or \emph{contiguous} or \emph{adjacent}) if there exist paths $\alpha, \beta \in P^ {\mathbb I}$ and an element $z \in P_2$ such that $d_1^0d_1^0z = \alpha (\length(\alpha))$, $d_1^1d_1^1z = \beta (0)$ and $\{\omega ,\nu \} = \{{\alpha \cdot (d_1^0z)_{\sharp} \cdot (d_2^ 1z)_{\sharp} \cdot \beta}, {\alpha \cdot (d_2^0z)_{\sharp} \cdot (d_1^ 1z)_{\sharp} \cdot \beta} \}$ (see   \cite{FajstrupGR, Goubault}). The \emph{dihomotopy} relation, denoted by $\sim$, is the equivalence relation generated by elementary dihomotopy.

The \emph{geometric realization} of a precubical set $P$ is the quotient space \[|P|=(\coprod _{n \geq 0} P_n \times [0,1]^n)/\sim\] where the sets $P_n$ are given the discrete topology and the equivalence relation is given by
\[(d^k_ix,u) \sim (x,\delta_i^k(u)), \quad  x \in P_{n+1},\; u\in [0,1]^n,\; i \in  \{1, \dots, n+1\},\; k \in \{0,1\}\] (see \cite{FahrenbergThesis, FajstrupGR,  GaucherGoubault, Goubault}). Here, the continuous maps $\delta_i^ k \colon [0,1]^ n \to [0,1]^ {n+1}$ are defined by $\delta_i^ k(u_1, \dots, u_n) = (u_1, \dots, u_{i-1},k,u_i, \dots, u_n)$.

The \emph{star} of an element $x$ of a precubical set $P$ is the set 
$$\star_P (x) = \{y \in P : x \in y_{\sharp}(\llbracket 0,1\rrbracket ^{\otimes \deg(y)})\}.$$
If no confusion is possible, we may omit the subscript and write $\star(x)$ instead of $\star_P(x)$. The graded set $P\setminus \star(x)$ is a precubical subset of $P$. We say that a face $d^k_ix$ of a cube $x$ of $P$ is \emph{free} if $\star(d^k_ix) = \{x,d^k_ix\}$. If $x$ is a regular cube of $P$ with free face $d^k_ix$, then the inclusion $|P \setminus \star(d^k_ix)| \hookrightarrow |P|$ is a homotopy equivalence.

\subsection{Higher-dimensional automata} \label{HDAdef}

Let $M$ be a monoid. A \emph{higher-dimensional automaton over} $M$ (abbreviated $M\mbox{-}$HDA or simply HDA) is a tuple $\A = (P,I,F, \lambda)$ where $P$ is a precubical set, $I \subseteq P_0$ is a set of \emph{initial states}, $F \subseteq P_0$ is a set of \emph{final} of \emph{accepting  states}, and $\lambda \colon P_1 \to M$ is a map, called the \emph{labeling function}, such that $\lambda (d_i^0x) = \lambda (d_i^1x)$ for all $x \in P_2$ and $i \in \{1,2\}$ (cf. \cite{vanGlabbeek, weakmor, hgraph}). A \emph{morphism} from an $M$-HDA $\B = (Q,J, G,\mu)$ to an $M$-HDA $\A = (P,I, F,\lambda)$ is a morphism of precubical sets  $f\colon Q \to P$  such that $f(J) \subseteq I$, $f(G) \subseteq F$ and  $\lambda(f(x)) = \mu(x)$ for all $x \in Q_1$.

An \emph{execution} in an $M$-HDA   $\A = (P,I, F,\lambda)$ is a finite or infinite path that starts in an initial state. A finite execution is \emph{successful} if it terminates in a final state.

The \emph{extended labeling function} of an $M$-HDA $\A = (P,I,F, \lambda)$ is the map $\overline{\lambda} \colon P^{\mathbb I} \to M$ defined as follows: If $\omega = x_{1\sharp} \cdots x_{k\sharp}$ for a sequence $(x_1, \dots , x_k)$ of elements of $P_1$ such that $d_1^0x_{j+1} = d_1^1x_j$ $(1\leq j < k)$, then we set $\overline{\lambda} (\omega) = \lambda (x_1) \cdots \lambda (x_k)$; if $\omega$ is a path of length $0$, then we set  $\overline{\lambda} (\omega ) = 1$. The \emph{language accepted by} $\A$, $L(\A)$, is the set of labels of successful executions, i.e., $$L(\A) = \{\overline{\lambda} (\omega) : \omega \in P^{\mathbb I},\; \omega (0) \in I,\; \omega(\length(\omega)) \in F\}.$$ If $M$ is a finitely generated free monoid and $\A = (P,I,F, \lambda)$ is an $M$-HDA such that $P_0$ and $P_1$ are finite, then $L(\A)$ is both a rational and a recognizable (cf. e.g. \cite{Diekert, Sakarovitch}) subset of $M$.

An $M$-HDA $\A = (P.I.F,\lambda)$ is called \emph{accessible} if for every vertex $w$ there exists an initial state $v\in I$ such that $v\to_P w$. The $M$-HDA $\A$ is called \emph{coaccessible} if for every vertex $v$ there exists a final state $w \in F$ such that $v\to_P w$.

\subsection{Weak morphisms \cite{weakmor, hgraph}} 
A \emph{weak morphism} from a precubical set $Q$ to a precubical set $P$ is a continuous map $f\colon |Q| \to |P|$ such that the following two conditions hold:
\begin{enumerate}
\item For every vertex $v\in Q_0$ there exists a (necessarily unique) vertex $f_0(v)\in P_0$ such that $f([v,()]) = [f_0(v),()]$.
\item For all integers $n, k_1, \dots, k_n\geq 1$ and every morphism of precubical sets $\xi \colon \llbracket 0,{k_1} \rrbracket\otimes \cdots \otimes \llbracket 0,{k_n} \rrbracket \to Q$, there exist integers $l_1, \dots, l_n\geq 1$, a morphism of precubical sets $\chi \colon \llbracket 0,{l_1} \rrbracket\otimes \cdots \otimes \llbracket 0,{l_n} \rrbracket \to P$, and a homeo\-morphism 
\begin{eqnarray*}\lefteqn{\phi\colon  |\llbracket 0,{k_1} \rrbracket\otimes \cdots \otimes \llbracket 0,{k_n} \rrbracket| = [0,k_1] \times \cdots \times [0,k_n]}\\ &\to& |\llbracket 0,{l_1} \rrbracket\otimes \cdots \otimes \llbracket 0,{l_n} \rrbracket|= [0,l_1] \times \cdots \times [0,l_n] \end{eqnarray*} 
such that $f\circ |\xi| = |\chi|\circ \phi$ and $\phi$ is a dihomeomorphism, i.e., $\phi$ and $\phi^{-1}$ preserve the natural partial order of $\R^n$. 
\end{enumerate}
We remark that weak morphisms are stable under composition. Note also that the integers $l_1, \dots, l_n\geq 1$, the morphism of precubical sets $\chi$, and the dihomeomorphism $\phi$ in condition (2) above are unique and that $\phi$ is itself a weak morphism \cite[2.3.5]{weakmor}. In the case of the morphism of precubical sets $x _{\sharp} \colon \llbracket 0,1\rrbracket ^ {\otimes n}\to Q$ corresponding to an element $x \in Q_n$ $(n>0)$, we shall use the notation $R_x = \llbracket 0,{l_1} \rrbracket\otimes \cdots \otimes \llbracket 0,{l_n} \rrbracket$, $\phi_x = \phi$, and $x_{\flat} = \chi$. 

Consider a weak morphism of precubical sets $f \colon |Q| \to |P|$. Let $X$ be a precubical subset  of $Q$. Then there exists a unique precubical subset $f(X)$ of $P$ such that $f(|X|) = |f(X)|$ \cite{hgraph}. Let $\omega \colon \llbracket 0,k \rrbracket \to Q$ $(k \geq 0)$ be a path. If $k > 0$, we denote by $f^{\mathbb I}(\omega)$ the unique path $\nu \colon \llbracket 0,{l} \rrbracket \to P$ for which there exists a dihomeomorphism $\phi \colon |\llbracket 0,k \rrbracket| = [0,k] \to |\llbracket 0,{l} \rrbracket| = [0,l]$ such that $f\circ |\omega | = |\nu |\circ \phi$. If $k= 0$, $f^{\mathbb I}(\omega)$ is defined to be the path in $P$ of length $0$ given by $f^{\mathbb I}(\omega)(0) = f_0(\omega(0))$. Note that the path 
$f^{\mathbb I}(\omega )$ leads from $f_0(\omega(0))$ to $f_0(\omega(k))$ and that the map $f^{\mathbb I}\colon Q^{\mathbb I} \to P^{\mathbb I}$ is  compatible with composition, concatenation, and dihomotopy \cite[2.5, 3.6]{weakmor}.
\begin{sloppypar}
A \emph{weak morphism} from an $M$-HDA $\B = {(Q,J, G,\mu)}$  to an $M$-HDA $\A = (P,I, F,\lambda)$  is a weak morphism $f\colon |Q| \to |P|$  such that $f(J) \subseteq I$, $f(G) \subseteq F$, and  $\overline{\lambda}\circ f^{\mathbb I} = \overline{\mu}$. One easily checks that the existence of a weak morphism from an $M$-HDA $\B$ to an $M$-HDA $\A$ implies that $L(\B) \subseteq L(\A)$ \cite[2.6.2]{weakmor}. One also easily verifies the following fact:
\end{sloppypar}

\begin{prop} \label{coacc}
Let $f$ be a weak morphism from an $M$-HDA $\B$ to an $M$-HDA $\A$. If $\B$ is (co)accessible and for every vertex ${a \in P_0}$ there exists a vertex $b\in Q_0$ such that $f_0(b) \to_P a$ $(a\to_P f_0(b))$, then $\A$ is (co)accessible. \hfill $\square$
\end{prop}

\section{Weak regularity} \label{wr}

Many interesting precubical sets---such as, for example, the directed circle, which consists of one vertex and one edge---are not regular. In many cases, they are, however, still weakly regular (as defined below). Properties of regular precubical sets and HDAs can often be generalized to weakly regular ones.

\subsection{Weakly regular elements} We say that an element $x$ of degree $n$ of a precubical set is \emph{weakly regular} if the restrictions of $x_{\sharp}$ to the graded subsets $\llbracket 0,1 \llbracket^ {\otimes n}$ and $\rrbracket 0,1 \rrbracket^ {\otimes n}$ are injective. Every element of degree $\leq 1$ is weakly regular. We remark that if $n > 0$, $\llbracket 0,1 \llbracket^ {\otimes n}= \star((0,\dots,0))$ and $\rrbracket 0,1 \rrbracket^ {\otimes n} = \star((1,\dots,1))$.

\begin{prop} \label{wreg2}
Let $P$ be a precubical set and $x \in P_2$ be an element. Then $x$ is weakly regular if and only if $d_1^0x \not= d_2^0x$ and $d_1^1x \not= d_2^1x$. 
\end{prop}

\begin{proof}
For $k \in \{0,1\}$, $\star((k,k)) = \{(k,k), (k,[0,1]), ([0,1],k), ([0,1],[0,1])\}.$ The restriction of $x_{\sharp}$ to $\star((k,k))$ is injective if and only if $d_1^kx = x_{\sharp}(k,[0,1]) \not= x_{\sharp}([0,1],k) = d_2^ kx.$
\end{proof}

\subsection{The edge $\boldsymbol{e^ k_ix}$} \label{eki}

Let $x$ be an element of degree $n > 0$ of a precubical set $P$, and let $k\in\{0,1\}$ and $i \in \{1, \dots, n\}$. We define the element $e^ k_ix \in P_1$ by 
$$e^ k_ix = \left\{\begin{array}{ll} x,& n = 1,\\d_1^{1-k} \cdots d_{i-1}^ {1-k}d_{i+1}^{1-k}\cdots d_n^{1-k}x, & n > 1. \end{array} \right.$$ Note that the element $e^0_ix$ is an edge of $x$ leading from the final vertex of the face $d^ 0_ix$ to the final vertex of $x$, i.e., we have $d_1^0e^ 0_ix = d_1^1\cdots d_1^1d^0_ix$ and $d_1^1e^0_ix = d_1^ 1\cdots d_1^1x$. Similarly, $e^1_ix$ is an edge of $x$ leading from the initial vertex of $x$ to the initial vertex of the face $d^ 1_ix$, i.e., we have $d_1^ 0e^ 1_ix = d_1^ 0\cdots d_1^0x$ and $d_1^1e^ 1_ix = d_1^ 0\cdots d_1^0d_i^ 1x$.

\subsection{Weakly regular precubical sets and HDAs} 
\begin{sloppypar}
A precubical set or an $M\mbox{-}$HDA is said to be \emph{weakly regular} if all its elements are weakly regular. 
\end{sloppypar}

\begin{prop}\label{carwreg}
A precubical set $P$ is weakly regular if and only if all elements of degree $2$ are weakly regular. 
\end{prop}

\begin{proof}
The ``only if'' part is obvious. Suppose that all elements of degree $2$ are weakly regular. Then all elements of degrees $0$, $1$, and $2$ are weakly regular. Consider an element $x$ of degree $n > 2$. Suppose that the restriction of $x_{\sharp}$ to $\llbracket 0,1 \llbracket^ {\otimes n}$ is not injective.  Let $a$ and $b$  be elements of $\llbracket 0,1 \llbracket^ {\otimes n}$ such that $a \not= b$ and $x_{\sharp}(a) = x_{\sharp}(b)$. Then $a$ and $b$ are of the same degree $0<k<n$. Write $a = (a_1, \dots, a_n)$ and $b = (b_1, \dots, b_n)$ and consider the uniquely determined sequences $1\leq i_1 < \cdots < i_k \leq n$ and $1\leq j_1 < \cdots < j_k \leq n$ such that $$a_i = \left\{\begin{array}{ll}[0,1], & i \in \{i_1,\dots ,i_k\},\\0, & i \notin \{i_1,\dots ,i_k\} \end{array}\right. \quad \mbox{and} \quad b_i = \left\{\begin{array}{ll}[0,1], & i \in \{j_1,\dots ,j_k\},\\0, & i \notin \{j_1,\dots ,j_k\}. \end{array}\right.$$ For $r\in \{1, \dots, k\}$, 
$$(e_r^1a)_i = \left\{\begin{array}{ll}[0,1], & i = i_r,\\0, & i \not= i_r \end{array}\right. \quad \mbox{and} \quad (e_r^1b)_i = \left\{\begin{array}{ll}[0,1], & i = j_r,\\0, & i \not= j_r. \end{array}\right.$$  Since $a\not =b$, there exists an $r \in \{1, \dots, k\}$ such that $i_r\not= j_r$. Hence $e_r^1a \not= e_r^1b$. We may suppose that $i_r < j_r$. Consider the element $c \in \llbracket 0,1 \llbracket^ {\otimes n}_2$ defined by $$c_i = \left\{\begin{array}{ll}[0,1], & i \in \{i_r,j_r\},\\0, & i \notin \{i_r,j_r\}. \end{array}\right.$$ Then $d_1^0c = e_r^1b$ and $d_2^0c = e_r^1a$. Moreover, $x_{\sharp}(c)$ is weakly regular. Therefore $e_r^1x_{\sharp}(b) = x_{\sharp}(e_r^1b) = x_{\sharp}(d_1^0c)= d_1^0x_{\sharp}(c) \not= d_2^0x_{\sharp}(c) = x_{\sharp}(d_2^0c) = x_{\sharp}(e_r^1a) = e_r^1x_{\sharp}(a)$. This is impossible because $x_{\sharp}(a) = x_{\sharp}(b)$. It follows that the restriction of $x_{\sharp}$ to $\llbracket 0,1 \llbracket^ {\otimes n}$ is injective. An analogous argument shows that the restriction of $x_{\sharp}$ to $\rrbracket 0,1 \rrbracket^ {\otimes n}$ is injective. Thus $x$ is weakly regular.
\end{proof}

\subsection{Weak morphisms and weak regularity} 

One easily sees that if the target of a morphism of precubical sets is regular, then so is the source. Here, we show the corresponding statement for weak morphisms and weak regularity.

\begin{prop} \label{weakreg}
Let $f\colon |Q| \to |P|$ be a weak morphism of precubical sets such that $P$ is weakly regular. Then $Q$ is weakly regular.
\end{prop}

\begin{proof}
Consider an element $x \in Q_2$. We show that $d^0_1x \not= d^0_2x$. The proof of the fact that $d^1_1x \not= d^1_2x$ is analogous. Suppose that $d^0_1x = d^0_2x$. Write $R_x = \llbracket 0,k\rrbracket \otimes \llbracket 0, l\rrbracket$. By \cite[2.1.2]{weakmor}, we have either $\phi_x(\{0\}\times [0,1]) = \{0\}\times [0,l]$ or $\phi_x(\{0\}\times [0,1]) = [0,k]\times \{0\}$. We may suppose that $\phi_x(\{0\}\times [0,1]) = \{0\}\times [0,l]$. Then $\phi_x([0,1]\times \{0\}) = [0,k]\times \{0\}$. Consider the element $y = x_{\flat}([0,1],[0,1])$ and form the following commutative diagram:
\[\xymatrix{
[0,1] \ar[d]_{\approx} & [0,l] \ar[d]_{\approx} & [0,1] \ar[d]^{\approx} \ar@{_{(}->}[l]\\
\{0\}\times [0,1] \ar@{^{(}->}+<0ex,-2ex>;[d] \ar[r]^{\phi_x}_{}
&  \{0\}\times [0,l]  \ar@{^{(}->}+<0ex,-2ex>;[d]  
& \{0\}\times [0,1] \ar@{^{(}->}+<0ex,-2ex>;[d] \ar@{_{(}->}[l]_{} 
\\ 
[0,1]^2 \ar[d]^{}_{|x_{\sharp}|} \ar[r]^{\phi_x}_{}
&  [0,k]\times [0,l]  \ar[d]^{}_{|x_{\flat}|}  
& [0,1]^2 \ar[dl]^{|y_{\sharp}|}_{}\ar@{_{(}->}[l]
\\
|Q| \ar[r]^{}_{f} 
& |P|. 
& 
}\]
The column on the left hand side is $|(d^0_1x)_{\sharp}|$. The middle column is thus $|f^{\mathbb I}((d^0_1x)_{\sharp})|$. The composition on the right hand side is $|(d^0_1y)_{\sharp}|$. Construct a similar diagram with $|(d^0_2x)_{\sharp}|$ on the left, $|f^{\mathbb I}((d^0_2x)_{\sharp})|$ in the middle, and $|(d^0_2y)_{\sharp}|$ on the right. Since $d^0_1x = d^0_2x$, we have $l = k$ and $f^{\mathbb I}((d^0_1x)_{\sharp}) = f^{\mathbb I}((d^0_2x)_{\sharp})$. It follows that $|(d^0_1y)_{\sharp}| = |(d^0_2y)_{\sharp}|$ and hence that $d^0_1y = d^0_2y$. Since $P$ is weakly regular, this is impossible.
\end{proof}

\section{The trace category} \label{TC}

In this section, we define the trace category of an HDA and establish criteria for its invariance under cube collapses. We also show that certain dihomotopy invariant properties are preserved by weak morphisms that induce full functors of trace categories.

\subsection{The trace category of an HDA}
The \emph{fundamental category} of a precubical set $P$ is the category $\vec \pi_1(P)$ whose objects are the vertices of $P$ and whose morphisms are the dihomotopy classes of paths in $P$ (cp. \cite{Goubault, GrandisBook}). Given a weak morphism of precubical sets $f\colon |Q| \to |P|$, the functor $f_*\colon \vec \pi_1(Q) \to \vec \pi_1(P)$ is defined  by $f_*(v) = f_0(v)$ $(v\in Q_0)$ and  $f_*([\omega]) = [f^{\mathbb I}(\omega)]$ $(\omega \in Q^{\mathbb I})$. 

A vertex $v$ of a precubical set $P$ is said to be \emph{maximal (minimal)} if there is no element $x \in P_1$ such that $d_1^0x = v$ $(d_1^1x = v)$. The sets of maximal and minimal elements of $P$ are denoted by $m_0(P)$ and $m_1(P)$ respectively. The \emph{trace category} of an $M$-HDA $\A = (P,I,F,\lambda)$, $TC(\A)$, is the full subcategory of $\vec \pi_1(P)$ generated by $I\cup F\cup m_0(P) \cup m_1(P)$ \cite{weakmor}. The definition of the trace category of an HDA is a variant  of Bubenik's definition of the fundamental bipartite graph of a d-space \cite{BubenikExtremal}. Note that if $f$ is a weak morphism from an $M$-HDA $\B = (Q,J, G,\mu)$ to an $M\textrm{-}$HDA $\A = (P,I, F,\lambda)$ such that $$f_0(m_0(Q) \cup m_1(Q)) \subseteq I\cup F\cup m_0(P) \cup m_1(P),$$ then the functor $f_*\colon \vec \pi_1(Q) \to \vec \pi_1(P)$ restricts to a functor $f_*\colon TC(\B) \to TC(\A)$. 

Given an $M$-HDA $\A = (P,I, F,\lambda)$ and a precubical subset $Q \subseteq P$ such that $Q_0 = P_0$, one has $I \cup F \subseteq Q_0$ and can form the $M$-HDA $\B = (Q,I,F,\lambda|_{Q_1})$. If $Q_{\leq 1} = P_{\leq 1}$, then $m_j(Q) = m_j(P)$ $(j = 0,1)$ and the inclusion $\iota \colon |Q| \hookrightarrow |P|$  induces a full functor $\iota_*\colon TC(\B) \to TC(\A)$ that is the identity on objects. We clearly have the following proposition, which in particular implies that the trace category is invariant under elementary collapses of cubes of dimension $\geq 4$:

\begin{prop} 
If $Q_{\leq 2} = P_{\leq 2}$, then the functor $\iota_*\colon TC(\B) \to TC(\A)$ is an isomorphism.
\hfill $\square$
\end{prop}
\subsection{Divisibility of dihomotopy classes} Let $P$ be a precubical set. We say that a dihomotopy class $[\omega] \in \vec \pi_1 (P)$ is \emph{left} or \emph{$0$-divisible} by a dihomotopy class $[\alpha]$ if there exists a dihomotopy class $[\beta]$ such that $[\omega] = [\alpha]\cdot [\beta]$. Similarly, $[\omega]$ is \emph{right} or \emph{$1$-divisible} by  $[\alpha]$ if there exists a dihomotopy class $[\beta]$ such that $[\omega] = [\beta]\cdot [\alpha]$. 

\subsection{Invariance of the trace category under elementary $\boldsymbol{3}$-cube collapses} Let $\A = (P,I, F,\lambda)$ be an $M$-HDA, and let $x$ be a regular $3$-cube with free face $d^ k_ix$. Consider the precubical subset $Q = P \setminus \star(d^ k_ix)$ of $P$ and the $M$-HDA $\B = (Q,I,F,\lambda|_{Q_1})$. Since $Q_{\leq 1} = P_{\leq 1}$, the inclusion $\iota \colon |Q| \hookrightarrow |P|$ induces a full functor $\iota_*\colon TC(\B) \to TC(\A)$ that is the identity on objects. Recall from \ref{eki} that $e^k_ix$ is the edge of $x$ such that $d^ k_1e^k_ix = d^{1-k}_1d^{1-k}_1d^k_ix$ and $d^ {1-k}_1e^k_ix = d^{1-k}_1d^{1-k}_1d^{1-k}_1x$.

\begin{theor} \label{TC3}
\begin{sloppypar}
Suppose that for every path $\omega \in P^ {\mathbb I}$ with $\omega (k\cdot \length(\omega)) = d^ {1-k}_1d^{1-k}_1d^k_ix$ and $\omega ((1-k)\cdot \length(\omega)) \in {I \cup F \cup m_0(P) \cup m_1(P)}$, $[\omega]$ is $k$-divisible by $[(e^k_ix)_{\sharp}]$. Then $\iota_*\colon TC(\B) \to TC(\A)$ is an isomorphism.
\end{sloppypar}
\end{theor}

\begin{proof}
We have to show that $\iota_*$ is faithful. We only treat the case $k=1$ and leave the analogous case $k=0$ to the reader. We show that any two paths starting in a vertex of $I\cup F \cup m_0(P)\cup m_1(P)$ that are dihomotopic in $P$ are dihomotopic in $Q$. We proceed by induction on the length of the paths. It is clear that the assertion holds for paths of length $0$. Consider two paths $\omega, \nu \in P^{\mathbb I}$ of length $l > 0$ beginning in a vertex of ${I \cup F \cup m_0(P) \cup m_1(P)}$ that are dihomotopic in $P$, and suppose that the assertion holds for any two such paths of length $< l$. We may suppose that $\omega$ and $\nu$ are elementarily dihomotopic in $P$. Then there exist paths $\alpha, \beta \in P^ {\mathbb I} = Q^ {\mathbb I}$ and an element $z \in P_2$ such that $d_1^0d_1^0z = \alpha (\length(\alpha))$, $d_1^1d_1^1z = \beta (0)$ and $\{\omega ,\nu \} = \{{\alpha \cdot (d_1^0z)_{\sharp} \cdot (d_2^ 1z)_{\sharp} \cdot \beta}, {\alpha \cdot (d_2^0z)_{\sharp} \cdot (d_1^ 1z)_{\sharp} \cdot \beta }\}$.
If $z\not=d_i^1x$, then $z \in Q_2$ and $\omega$ and $\nu$ are elementarily dihomotopic in $Q$. Suppose that $z=d_i^1x$. By the hypothesis of the theorem, there exists a path $\gamma$ from $\alpha(0)$ to $d^0_1d^0_1d^0_1x$ such that $\alpha$ is dihomotopic in $P$ to $\gamma \cdot (e^1_ix)_{\sharp}$. By the inductive hypothesis, $\alpha \sim \gamma \cdot (e^1_ix)_{\sharp}$ in $Q$. Since $x$ is regular,  $(e^1_ix)_{\sharp}\cdot (d_1^0z)_{\sharp} \cdot (d_2^ 1z)_{\sharp}\sim (e^1_ix)_{\sharp} \cdot (d_2^0z)_{\sharp} \cdot (d_1^ 1z)_{\sharp}$ in $Q$. Hence  $\gamma \cdot (e^1_ix)_{\sharp}\cdot (d_1^0z)_{\sharp} \cdot (d_2^ 1z)_{\sharp} \cdot \beta \sim \gamma \cdot (e^1_ix)_{\sharp} \cdot (d_2^0z)_{\sharp} \cdot (d_1^ 1z)_{\sharp}\cdot \beta$ in $Q$. Thus $\omega \sim \nu$ in $Q$.
\end{proof}

\subsection{Invariance under elementary $\boldsymbol{2}$-cube collapses}
Let $\A = (P,I, F,\lambda)$ be an $M$-HDA, and let $x$ be a regular $2$-cube with free face $d^ k_ix$. Consider the precubical subset $Q = P \setminus \star(d^ k_ix)$ of $P$ and the $M$-HDA $\B = (Q,I,F,\lambda|_{Q_1})$. 

\begin{sloppypar}
\begin{theor} \label{TC2}
Suppose that 
\begin{enumerate}[(i)]
\item there exists an edge $y \in P_1\setminus \{d^k_ix\}$ such that $d^{1-k}_1y = d^{1-k}_1d^k_ix$; and
\item for every path $\omega \in Q^ {\mathbb I}$ with $\omega (k\cdot \length(\omega)) = d^ {1-k}_1d^k_ix$ and ${\omega ((1-k)\cdot \length(\omega))} \in {I \cup F \cup m_0(P) \cup m_1(P) \cup \{d^k_1d^ k_1x\}}$, $[\omega]$ is uniquely $k$-divisible in $\vec \pi_1(Q)$ by $[(e^k_ix)_{\sharp}]$.
\end{enumerate}
Then the inclusion $\iota \colon |Q| \hookrightarrow |P|$ induces an isomorphism $\iota_*\colon TC(\B) \to TC(\A)$.
\end{theor}

\begin{proof}
It follows from Lemma \ref{minmax} below that $\iota$ induces a functor $\iota_*\colon TC(\B) \to TC(\A)$ that is the identity on objects.

Define a functor $\rho \colon TC(\A) \to TC(\B)$ as follows: For vertices $v \in {I \cup F \cup m_0(P)\cup m_1(P)}$, set $\rho (v) = v$. For paths $\omega$ in $P$ beginning and ending in $I \cup F \cup  m_0(P)\cup m_1(P)$, set $\rho ([\omega]) = [\omega']$ where $\omega'$ is the path constructed in Lemma \ref{aprime} below. One easily checks that $\rho = \iota_*^{-1}$. 
\end{proof}
\end{sloppypar}

\begin{lem} \label{minmax}
Under the first hypothesis of Theorem \ref{TC2}, $m_0(Q) = m_0(P)$ and $m_1(Q) = m_1(P)$.
\end{lem}

\begin{proof}
Since $Q_0 = P_0$, we have $m_0(P) \subseteq m_0(Q)$ and $m_1(P) \subseteq m_1(Q)$. Consider a vertex $v\in m_k(Q)$. Since $d_1^kd_i^kx = d_1^kd_{3-i}^kx$ and $d_{3-i}^kx \in Q_1$, $v \not= d_1^kd_i^kx$. It follows that $v\in m_k(P)$ and hence that $m_k(Q) = m_k(P)$. Now consider a vertex $v \in P_0 = Q_0$ such that $v \notin m_{1-k}(P)$. Then there exists an element $a \in P_1$ such that $d_1^{1-k}a = v$. If $a \not= d_i^kx$, then $a \in Q_1$ and $v \notin m_{1-k}(Q)$. If $a = d_i^kx$, then condition (i) of Theorem \ref{TC2} ensures that there is an element $b \in Q_1$ such that $d_1^{1-k}b = d_1^{1-k}a = v$. Therefore also in this case $v \notin m_{1-k}(Q)$. Hence $m_{1-k}(Q) = m_{1-k}(P)$.
\end{proof}

\begin{lem} \label{aprime}
Under the second hypothesis of Theorem \ref{TC2}, for each path $\omega \in P^{\mathbb I}$ with ${\omega ((1-k)\cdot \length(\omega))} \in {I \cup F \cup m_0(P) \cup m_1(P)}$, there exists a path  $\omega'\in Q^{\mathbb I}$ such that 
$\omega' = \omega$ if $\omega \in Q^{\mathbb I}$, 
$\omega' \sim \omega$ in $P$, and
$\omega' \sim \nu'$ in $Q$ if $\omega \sim \nu$ in $P$.
\end{lem}

\begin{proof}
As before, we only treat the case $k=1$ and leave the analogous case $k=0$ to the reader. 
Write $\omega$ uniquely as a concatenation $$\omega = \omega_0\cdot (d^1_ix)_{\sharp}\cdot  \omega_1 \cdots \omega_{r_{\omega}-1}\cdot (d^1_ix)_{\sharp} \cdot \omega_{r_{\omega}}$$
where each $\omega_j$ is a path in $Q$. By the second hypothesis of Theorem \ref{TC2}, there exist paths $\bar \omega_j \in Q^ {\mathbb I}$ $(j< r_{\omega})$ such that $\bar \omega_j\cdot (e^1_ix)_{\sharp} = \bar \omega_j\cdot (d^0_{3-i}x)_{\sharp}\sim \omega_j$ in $Q$. Set
$$\omega' = \bar \omega_0\cdot (d^0_{i}x)_{\sharp}\cdot (d^1_{3-i}x)_{\sharp}\cdot  \bar \omega_1 \cdots \bar \omega_{r_{\omega}-1}\cdot (d^0_{i}x)_{\sharp}\cdot (d^1_{3-i}x)_{\sharp} \cdot \omega_{r_{\omega}}.$$
Since $x$ is regular, $\omega'\in Q^{\mathbb I}$. It is clear that $\omega'$ satisfies the first two properties of the statement. In order to show the remaining  property, let $\omega,\nu\in P^{\mathbb I}$ be elementarily dihomotopic paths beginning in a vertex of ${I \cup F \cup m_0(P) \cup m_1(P)}$. Then there exist paths $\alpha, \beta \in P^ {\mathbb I}$ and an element $z \in P_2$ such that $d_1^0d_1^0z = \alpha (\length(\alpha))$, $d_1^1d_1^1z = \beta (0)$ and $\{\omega ,\nu \} = \{{\alpha \cdot (d_1^0z)_{\sharp} \cdot (d_2^ 1z)_{\sharp} \cdot \beta}, {\alpha \cdot (d_2^0z)_{\sharp} \cdot (d_1^ 1z)_{\sharp} \cdot \beta }\}$. 

Suppose first that $z = x$. Then we may assume that for some $0\leq s<r_{\omega}$, $$\alpha \cdot (d^0_{3-i}x)_{\sharp} = \omega_0\cdot (d^1_ix)_{\sharp}\cdot  \omega_1 \cdots \omega_{s-1}\cdot (d^1_ix)_{\sharp}\cdot \omega_{s}$$
and
$$\beta = \omega_{s+1}\cdot (d^1_ix)_{\sharp}\cdot  \omega_{s+2} \cdots \omega_{r_{\omega}-1}\cdot (d^1_ix)_{\sharp}\cdot \omega_{r_{\omega}}.$$ Then there exists a path $\phi\in Q^{\mathbb I}$ such that $\omega_s = \phi\cdot (d^0_{3-i}x)_{\sharp}$, 
$$\alpha  = \omega_0\cdot (d^1_ix)_{\sharp}\cdot  \omega_1 \cdots \omega_{s-1}\cdot (d^1_ix)_{\sharp}\cdot \phi,$$
and
$$\nu = \omega_0\cdot (d^1_ix)_{\sharp} \cdots  \omega_{s-1}\cdot (d^1_ix)_{\sharp}\cdot \phi\cdot (d^0_ix)_{\sharp}\cdot(d^1_{3-i}x)_{\sharp}\cdot \omega_{s+1}\cdot (d^1_ix)_{\sharp}\cdots   (d^1_ix)_{\sharp}\cdot \omega_{r_{\omega}}.$$ Thus, $r_{\nu} = r_{\omega}-1$ and 
$$\nu_j = \left\{\begin{array}{ll}
\omega_j, & j< s,\\
\phi\cdot (d^0_ix)_{\sharp}\cdot(d^1_{3-i}x)_{\sharp}\cdot \omega_{s+1}, & j = s,\\
\omega_{j+1}, & j > s.
\end{array} \right.$$ We have $\phi \sim \bar \omega_s$ in $Q$ and therefore $\nu_s \sim \bar \omega_s \cdot (d^0_ix)_{\sharp}\cdot (d^1_{3-i}x)_{\sharp}\cdot  \omega_{s+1} $ in $Q$. If $s < r_{\omega}-1$, this implies that $\bar \nu_s \sim \bar \omega_s \cdot (d^0_ix)_{\sharp}\cdot (d^1_{3-i}x)_{\sharp}\cdot \bar \omega_{s+1}$ in $Q$. It follows that $\nu' \sim \omega'$ in $Q$.

Suppose now that $z \not= x$. Then $r_{\omega} = r_{\nu}$ and there exists an index $s$ such that $\omega_j = \nu_j$ for all $j \not= s$ and $\omega_s$ and $ \nu_s$ are elementarily dihomotopic in $Q$. This implies that $\bar \omega_j \sim \bar \nu_j$ in $Q$ for all $j$ and hence that $\omega' \sim \nu'$ in $Q$.
\end{proof}

\subsection{The dihomotopy cancellation property} \label{cancel}

The uniqueness requirement in a divisibility condition as the one considered in Theorem \ref{TC2} is fulfilled if the path representing the divisor class has the following dihomotopy cancellation property: We say that a path $\gamma$ in a precubical set has the \emph{right dihomotopy cancellation property} if for all paths $\omega$ and $\nu$ ending in $\gamma (0)$,
$$[\omega]\cdot [\gamma] = [\nu]\cdot [\gamma] \Rightarrow [\omega] = [\nu].$$
Similarly, one defines the \emph{left dihomotopy cancellation property}. We remark that a path whose dihomotopy class is weakly invertible or a Yoneda morphism in the sense of \cite{Components, Components2} has both the right and left dihomotopy cancellation properties. 

The next proposition gives two sufficient conditions for the right dihomotopy cancellation property. The  corresponding statement for the left dihomotopy cancellation property also holds.

\begin{prop} \label{dihocancel}
Let $\gamma $ be a path in a precubical set such that no dihomotopic path starts with a back face of a $2$-cube or no edge ending in $\gamma(0)$ is a front face of a $2$-cube. Then $\gamma$ has the right dihomotopy cancellation property.
\end{prop}

\begin{proof}
Consider first a path $\phi$ ending in $\gamma(0)$, a path $\psi \sim \gamma$, and a path $\rho$ that is elementarily dihomotopic to $\phi\cdot \psi$. We show that there exists a factorization $\rho = \phi'\cdot \psi'$ with $\phi \sim \phi'$ and $\psi \sim \psi'$. Consider paths $\alpha, \beta$, a $2$-cube $z$, and $i\in \{1,2\}$ such that $d_1^0d_1^0z = \alpha (\length(\alpha))$, $d_1^1d_1^1z = \beta (0)$, $\phi\cdot \psi = \alpha \cdot (d_i^0z)_{\sharp} \cdot (d_{3-i}^ 1z)_{\sharp} \cdot \beta$, and $\rho= \alpha \cdot (d_{3-i}^0z)_{\sharp} \cdot (d_i^ 1z)_{\sharp} \cdot \beta$. Since $\psi$ does not start with the edge $d_{3-i}^ 1z$ or $\phi$ does not terminate with the edge $d^0_iz$, there exists either a path $\beta'$ such that $\phi = \alpha \cdot (d_i^0z)_{\sharp} \cdot (d_{3-i}^ 1z)_{\sharp} \cdot \beta'$ and $\beta = \beta'\cdot \psi$ or a path $\alpha'$ such that $\psi = \alpha' \cdot (d_i^0z)_{\sharp} \cdot (d_{3-i}^ 1z)_{\sharp} \cdot \beta$ and $\alpha = \phi\cdot \alpha'$. In the first case, $\rho= \alpha \cdot (d_{3-i}^0z)_{\sharp} \cdot (d_i^ 1z)_{\sharp} \cdot \beta' \cdot \psi$ and $\phi$ is elementarily dihomotopic to $\alpha \cdot (d_{3-i}^0z)_{\sharp} \cdot (d_i^ 1z)_{\sharp} \cdot \beta'$. In the second case, $\rho= \phi \cdot \alpha' \cdot (d_{3-i}^0z)_{\sharp} \cdot (d_i^ 1z)_{\sharp} \cdot \beta$ and $\psi$ is elementarily dihomotopic to $\alpha' \cdot (d_{3-i}^0z)_{\sharp} \cdot (d_i^ 1z)_{\sharp} \cdot \beta$. We therefore obtain a factorization $\rho = \phi'\cdot \psi'$ with $\phi \sim \phi'$ and $\psi \sim \psi'$, as required.

Let now $\omega$ and $\nu$ be paths ending in $\gamma(0)$ such that $[\omega]\cdot [\gamma] = [\nu]\cdot [\gamma]$. Let $\rho_0, \dots, \rho_r$ be paths such that $\rho_0 = \omega \cdot \gamma$, $\rho_r = \nu \cdot \gamma$, and for all $j \in \{0, \dots, r-1\}$, $\rho_j$ is elementarily dihomotopic to $\rho_{j+1}$. By what we have shown above, there exist paths $\omega_0, \dots, \omega_r$, $\gamma_0, \dots, \gamma_r$ such that $\omega_0 = \omega$, $\omega_r = \nu$, $\gamma_0 = \gamma_r = \gamma$, $\rho_j = \omega_j \cdot \gamma_j$, $\omega_j \sim \omega_{j+1}$, and $\gamma_j \sim \gamma_{j+1}$. In particular, $\omega \sim \nu$.
\end{proof}

\begin{ex} \label{excancel}
Consider a regular $2$-cube $x$ with free face $d^1_ix$ in a weakly regular $M$-HDA $\A = (P,I,F,\lambda)$, and suppose that the only edge ending in $d^0_1d^1_ix$ is $e^1_ix = d^0_{3-i}x$. Then $(e^1_ix)_{\sharp}$ has the right dihomotopy cancellation property both in $P$ and in $P \setminus \star(d^1_ix)$. Indeed, in both cases, the only path dihomotopic to $(e^1_ix)_{\sharp}$ is $(e^1_ix)_{\sharp}$ itself, and $e^1_ix$ is not the back face of any $2$-cube. Since every non-constant path ending in $d^0_1d^1_ix$ terminates with $(e^1_ix)_{\sharp}$, it follows that the second condition of Theorem \ref{TC2} is satisfied if $d^0_1d^1_ix \notin I\cup F$. 
\end{ex}

\subsection{Invariance of the trace category under vertex star collapses}

Let $x$ be a regular cube of degree $n\geq 2$ of an an $M$-HDA $\A = (P,I,F,\lambda)$, and let $k_1, \dots, k_n \in\{0,1\}$ such that at least one $k_i = 0$, at least one $k_i = 1$, and $\star(d^{k_n}_1\cdots d^{k_1}_1x) \subseteq x_{\sharp}(\llbracket 0,1 \rrbracket ^{\otimes n})$. Suppose that $d^{k_n}_1\cdots d^{k_1}_1x \notin I \cup F$. Consider the  precubical subset $Q = P \setminus \star(d^{k_n}_1\cdots d^{k_1}_1x)$ of $P$ and the $M$-HDA $\B = (Q, I, F, \lambda |_{Q_1})$, and let $\iota$ denote the inclusion $|Q| \hookrightarrow |P|$. The following result has recently been established by Misamore \cite{Misamore} in a slightly different setting:

\begin{theor} \label{omegaprime}
The functor $\iota_* \colon \vec \pi_1(Q) \to \vec \pi_1(P)$ is fully faithful. \hfill $\square$
\end{theor}
 
\begin{lem} \label{mj}
For $j = 0,1$, $m_j(Q) = m_j(P)$.
\end{lem}

\begin{proof}
It is clear that $m_j(P) \subseteq m_j(Q)$ for $j =0,1$. Consider an element $v\in Q_0$ that is not maximal in $P$. Then there exists an edge $y\in P_1$ such that $d^0_1y = v$. If $y \in Q_1$, then $v$ is not maximal in $Q$. Suppose that $y\notin Q_1$. Then $d^1_1y = d^{k_n}_1\cdots d^{k_1}_1x$. Let $\omega$ be a path in $P$ from $d^{k_n}_1\cdots d^{k_1}_1x$ to $d^1_1 \cdots d^1_1x$. By Theorem \ref{omegaprime}, the path $y_{\sharp}\cdot \omega$ corresponds to a non-constant path in $Q$ that begins in $v$. Hence $v$ is not maximal in $Q$ in this case, too. It follows that $m_0(Q) \subseteq m_0(P)$. An analogous argument shows that $m_1(Q) \subseteq m_1(P)$.
\end{proof}

Theorem \ref{omegaprime} and Lemma \ref{mj} immediately imply:

\begin{theor} \label{TCV}
The functor $\iota_* \colon \vec \pi_1(Q) \to \vec \pi_1(P)$  restricts to an isomorphism $\iota_* \colon TC(\B) \to TC(\A)$. \hfill $\square$
\end{theor}

\subsection{Dihomotopy invariant properties} \label{dihoinv}

Let $\A = (P,I, F,\lambda)$ be an $M$-HDA. Given a subset $L \subseteq M$, we say that $\A$ \emph{has property} $L$ if $L(\A) \subseteq L$. We note that if there exists a weak morphism from an $M$-HDA $\B$ to $\A$ and $\A$ has property $L$, then $\B$ also has property $L$. We say that $L$ is an \emph{$\A$-dimomotopy invariant property} if for any two dihomotopic paths $\omega , \nu \in P^{\mathbb I}$, $\overline{\lambda}(\omega) \in L \Rightarrow \overline{\lambda}(\nu) \in L$. 

\begin{ex} \label{exproperty}
Suppose that $M$ is a free monoid on an alphabet $\Sigma$, $M = \Sigma ^*$, and that $\lambda(P_1) \subseteq \Sigma$. Let us say that two elements $a,b\in \Sigma$ are \emph{locally independent} if there exists a square $z \in P_2$ such that $\{\lambda(d^0_1z), \lambda(d^0_2z)\} = \{a,b\}$. Now consider two elements $a,b\in \Sigma$ that are \emph{not} locally independent. Then $$L = \Sigma^* \cdot\{a\} \cdot \Sigma^*\cdot\{b\}\cdot \Sigma^*$$ is an $\A$-dihomotopy invariant property. Indeed, let $\omega , \nu \in P^{\mathbb I}$ be dihomotopic paths, and suppose that $\overline{\lambda}(\omega) \in L$. We have to show that $\overline{\lambda}(\nu) \in L$. We may suppose that $\omega$ and $\nu$ are elementarily dihomotopic. Write $\omega = x_{1\sharp}\cdots x_{k\sharp}$. Then $\overline{\lambda}(\omega) = \lambda(x_{1})\cdots \lambda(x_{k})$ and  there exist indices $1 \leq r < s \leq k$ such that $\lambda(x_r) = a$ and $\lambda(x_s) = b$. We may suppose that there exist an element $z\in P_2$ and an index $1 \leq i < k$ such that $x_i = d^0_1z$, $x_{i+1} = d^1_2z$, and $$\nu = x_{1\sharp}\cdots x_{(i-1)\sharp}\cdot (d^0_2z)_{\sharp}\cdot (d^1_1z)_{\sharp}\cdot x_{(i+2)\sharp}\cdots x_{k\sharp}. $$
Then $$\overline{\lambda}(\nu) = \lambda(x_1)\cdots \lambda(x_{i-1})\cdot \lambda(x_{i+1})\cdot \lambda(x_i)\cdot \lambda(x_{i+2})\cdots \lambda(x_{k}). $$
Since $a$ and $b$ are not locally independent, $\{\lambda(x_i),\lambda(x_{i+1})\} = \{\lambda(d^0_1z), \lambda(d^0_2z)\} \not= \{a,b\}$. Thus $\{r,s\} \not= \{i,i+1\}$. It follows that  $\overline{\lambda}(\nu) \in L$. 
\end{ex}

One easily establishes the following proposition:

\begin{prop}
The $\A$-dihomotopy invariant properties form a Boolean subalgebra of the power set $\mathfrak{P}(M)$. \hfill $\square$
\end{prop}

\begin{theor} \label{dihoprop}
\begin{sloppypar}
Let $f$ be a weak morphism from an $M$-HDA $\B = (Q,J,G,\mu)$ to $\A$ 
such that $f(J) = I$, $f(G) = F$,  ${f(m_0(Q) \cup m_1(Q))} \subseteq {I\cup F\cup m_0(P) \cup m_1(P)}$, and the induced functor $f_*\colon TC(\B) \to TC(\A)$ is full. Then $\A$ has a given $\A$-dihomotopy invariant property $L$ if and only if $\B$ has that property.
\end{sloppypar}
\end{theor}

\begin{proof}
The ``only if'' part is clear. Suppose that $\B$ has property $L$. Let $\omega \in P^{\mathbb I}$ be a path from a vertex in $I$ to a vertex in $F$. Since $f(J) = I$, $f(G) = F$, and $f_*\colon TC(\B) \to TC(\A)$ is a full functor, there exists a path $\alpha \in Q^{\mathbb I}$ from a vertex in $J$ to a vertex in $G$ such that $f_*[\alpha] = [f^{\mathbb I}(\alpha)] = [\omega]$. Since $L(\B) \subseteq L$, we have $\overline{\lambda}(f^{\mathbb I}(\alpha)) = \overline{\mu}(\alpha) \in L$. Since $\omega \sim f^{\mathbb I}(\alpha)$ and $L$ is $\A$-dihomotopy invariant, it follows that $\overline{\lambda}(\omega) \in L$. Thus $L(\A) \subseteq L$.
\end{proof}

\section{The homology graph} \label{HG}

The purpose of this section is to establish conditions under which the homology graph of a precubical set is invariant under cube collapses. We consider singular homology with coefficients in an arbitrary commutative unital ring, which we suppress from the notation.

\subsection{The homology graph of a precubical set} Let $P$ be a precubical set. We say that a homology class $\alpha \in H_*(|P|)$ \emph{points} to a homology class (of a possibly different degree) $\beta \in H_*(|P|)$ and write $\alpha \nearrow \beta$ if there exist precubical subsets $X, Y \subseteq P$ such that  $\alpha \in \im H_*(|X| \hookrightarrow |P|)$, $\beta \in \im H_*(|Y| \hookrightarrow |P|)$, and for all vertices $x \in X_0$ and $y \in Y_0$, $x \to_P y$. The \emph{homology graph} of $P$ is the directed graph whose vertices are the homology classes of $|P|$ and whose edges are given by the relation $\nearrow$.

\begin{theor} \label{morpoint} {\rm{\cite{hgraph}}} 
Consider a weak morphism of precubical sets $f\colon |Q| \to |P|$. Then the induced map $f_*\colon H_*(|Q|) \to H_*(|P|)$ is a graph homomorphism, i.e., for all homology classes $\alpha, \beta \in H_*(|Q|)$, $\alpha \nearrow \beta \Rightarrow f_*(\alpha) \nearrow f_*(\beta)$. If $f$ is a homeomorphism, then $f_*$ is a graph isomorphism, i.e., for all homology classes $\alpha, \beta \in H_*(|Q|)$, $\alpha \nearrow \beta \Leftrightarrow f_*(\alpha) \nearrow f_*(\beta)$.
\end{theor}

\subsection{Deformation of precubical subsets}

Let $C$ and $C'$ be precubical subsets of a precubical set $P$. We say that $C$ is \emph{deformable into} $C'$ if there exists a precubical subset $\hat C \subseteq P$ such that 
$C \subseteq \hat C \supseteq C'$ and  the inclusion $|C'|\hookrightarrow |\hat C|$ is a homotopy equivalence. 

\begin{lem} \label{deflemma} 
Consider a precubical set $P$ and a precubical subset $Q \subseteq P$, and suppose that the inclusion $\iota \colon |Q| \hookrightarrow |P|$ is a homotopy equivalence. Let $C$ and $C'$ be  precubical subsets of $P$ such that $C' \subseteq Q$ and $C$ is deformable into $C'$, and let $\alpha \in H_*(|Q|)$ be a homology class  such that $\iota_*(\alpha) \in \im H_*(|C| \hookrightarrow |P|)$. Then $\alpha \in \im H_*(|C'| \hookrightarrow |Q|)$. 
\end{lem}

\begin{proof}
Since $C$ is deformable into $C'$, there exists a precubical subset $\hat C \subseteq P$ such that 
$C \subseteq \hat C \supseteq C'$ and the inclusion $|C'|\hookrightarrow |\hat C|$ is a homotopy equivalence. Consider the following commutative diagram, in which all morphisms are induced by the inclusions:
$$\xymatrix{
H_*(|C'|) \ar[d]^{}_{} \ar[r]^{\cong}_{}
&  H_*(|\hat C|)  \ar[d]^{}_{}  
& H_*(|C|)  \ar [dl]^{}_{}\ar[l]_{} 
\\ 
H_*(|Q|)  \ar[r]^{\cong}_{\iota_*} 
& H_*(|P|). 
& 
}$$
The triangle on the right side shows that $\iota_*(\alpha) \in \im \, H_*(|\hat C| \hookrightarrow |P|)$. It follows that $\alpha \in \im \, H_*(|C'| \hookrightarrow |Q|)$. 
\end{proof}

\subsection{Invariance under elementary collapses} Let $P$ be a precubical set, and let $x$ be a regular element of degree $n\geq 2$ with free face $d^k_ix$.  Consider the precubical set $Q = P \setminus \star(d^ k_ix)$.  The inclusion $\iota\colon |Q| \hookrightarrow |P|$ is a homotopy equivalence, and so it induces  an isomorphism $\iota_*$ in homology.

\begin{theor} \label{di1lem6}
Suppose that every precubical subset $C$ of $P$ is deformable into a precubical subset $C'$ of $Q$ such that for every vertex $c'$ in $C'$, there exists a  vertex $c$ in $C$ such that  \begin{itemize}
\item when $k = 0$: $c \to_Q c'$ and $d^1_1\cdots d^1_1d^0_ix  \to_P c \Rightarrow d^1_1\cdots d^1_1x \to_Q c'$; and
\item when $k = 1$: $c' \to_Q c$ and $c \to_P d^0_1\cdots d^0_1d^1_ix \Rightarrow  c'\to_Q d^0_1\cdots d^0_1x$.\end{itemize} Then $\iota_*\colon H_*(|Q|) \to H_*(|P|)$ is a graph isomorphism.
\end{theor}

\begin{proof}
We prove the theorem only in the case $k = 1$. The proof in the case $k = 0$ is analogous and is left to the reader. Consider homology classes $\alpha, \beta \in H_*(|Q|)$. By Theorem \ref{morpoint}, it suffices to show that $\alpha \nearrow \beta$ if $\iota_*(\alpha) \nearrow \iota_*(\beta)$. Suppose that $\iota_*(\alpha) \nearrow \iota_*(\beta)$, and let $A, B \subseteq P$ be precubical subsets such that $\iota_*(\alpha) \in \im\, H_*(|A| \hookrightarrow |P|)$, $\iota_*(\beta) \in \im\, H_*(|B| \hookrightarrow |P|)$, and for all $a \in A_0$ and $b \in B_0$, $a\to_P b$. Let $A'$ be the precubical subset of $Q$ given by the hypothesis of the theorem. Let $\hat B$ be the precubical subset of $P$ given by
$$\hat B = \left\{\begin{array}{ll} B, & d^0_1\cdots d^0_1d^1_ix \notin B,\\
B \cup x_{\sharp}(\llbracket 0, 1\rrbracket^{\otimes n}), & d^0_1\cdots d^0_1d^1_ix \in B,
 \end{array}\right.$$ 
and let $B'$ be the precubical subset of $Q$ given by $B' = \hat B\cap Q$. Since $d_i^1x \in \hat B\Leftrightarrow x\in \hat B$, the inclusion $|B'| \hookrightarrow |\hat B|$ is a homotopy equivalence. Thus $B$ is deformable into $B'$. By Lemma \ref{deflemma}, we have $\alpha \in \im \, H_*(|A'| \hookrightarrow |Q|)$ and $\beta \in \im \, H_*(|B'| \hookrightarrow |Q|)$.

We show that there exists a path in $Q$ from any element of $A'_0$ to any element of $B_0 \cup (B'_0\setminus B_0)$ and hence to any element of $B'_0$. Consider first an element $a' \in A'_0$ and an element $b \in B_0$. Then there exists a path $\rho \in Q^{\mathbb I}$ from $a'$ to an element $a$ in $A_0$ such that $a \to_P d^0_1\cdots d^0_1d^1_ix \Rightarrow  a'\to_Q d^0_1\cdots d^0_1x$. Since $\iota_*(\alpha) \nearrow \iota_*(\beta)$, there exists a path $\sigma \in P^{\mathbb I}$ from $a$ to $b$. The concatenation $\rho \cdot \sigma$ is a path in $P$ from $a'$ to $b$. If $n \geq 3$, $\rho \cdot \sigma\in Q^{\mathbb I}$ because then $Q_1= P_1$. Suppose that $n=2$, and suppose that $\rho \cdot \sigma\notin Q^{\mathbb I}$. Then $\sigma\notin Q^{\mathbb I}$ and there exist paths $\omega \in P^{\mathbb I}$ and $\nu \in Q^ {\mathbb I}$ such that $\sigma = \omega \cdot (d^1_ix)_{\sharp}\cdot \nu$. It follows that there exists a path $\gamma  \in Q^{\mathbb I}$ from $a'$ to $d^0_1d^0_1x$. The concatenation $\gamma \cdot (d^0_ix)_{\sharp} \cdot (d^1_{3-i}x)_{\sharp} \cdot \nu$ is a path in $Q$ from $a'$ to $b$.

Consider now an element $a' \in A'_0$ and an element $b' \in B'_0 \setminus B_0$. Then we have $d_1^0\cdots d_1^0d_i^1x \in B_0$ and $b' \in x_{\sharp}(\llbracket 0,1\rrbracket^{\otimes n})$. By our hypotheses, there exists an element $a\in A_0$ such that $a'\to_Q a$ and $a \to_P d^0_1\cdots d^0_1d^1_ix \Rightarrow  a'\to_Q d^0_1\cdots d^0_1x$. Since $d_1^0\cdots d_1^0d_i^1x \in B_0$ and $\iota_*(\alpha) \nearrow \iota_*(\beta)$, we have $a \to_P d_1^0\cdots d_1^0d_i^1x$ and therefore $a' \to_Q d_1^0\cdots d_1^0x$. Since $b' \in x_{\sharp}(\llbracket 0,1\rrbracket^{\otimes n})$, we have  $d_1^0\cdots d_1^0x\to_Qb'$. It follows that $a'\to_Qb'$.
\end{proof}

\subsection{A more concrete invariance criterion} \label{moreconcrete}

Let $P$ be a weakly regular precubical set, and let $x$ be a regular element of degree $n\geq 1$ with free face $d^k_ix$. Suppose that there exists no element $y \in P_1\setminus \{e^k_ix\}$ such that $d^k_1y = d_1^{1-k} \cdots d_1^{1-k}d_i^kx$. Consider the precubical set $Q = P \setminus \star(d^k_ix)$. Then the inclusion $\iota\colon |Q| \hookrightarrow |P|$ is a homotopy equivalence and induces  an isomorphism in homology.

\begin{theor} \label{red1}
The map $\iota_*\colon H_*(|Q|) \to H_*(|P|)$ is a graph isomorphism.   
\end{theor}

\begin{proof}
\begin{sloppypar}
Again, we only prove the theorem in the case $k = 1$. We check the check the conditions of Theorem \ref{di1lem6}. Let $C$ be a precubical subset of $P$. We show that $C$ is deformable into a precubical subset $C'$ of $Q$. We   
define precubical subsets $\hat C$ of $P$ and  $C'$ of $Q$ by setting 
$$\hat C = \left\{\begin{array}{ll} C, & d_1^0\cdots d_1^0d_i^1x \notin C,\\
C\cup (e^1_ix)_{\sharp}(\llbracket 0,1\rrbracket), & d_1^0\cdots d_1^0d_i^1x \in C\; \mbox{and}\; d_i^1x \notin C,\\ 
C \cup x_{\sharp}(\llbracket 0, 1\rrbracket^{\otimes n}), & d_i^1x \in C
 \end{array}\right.$$ 
and 
$$C' = \left\{\begin{array}{ll} \hat C\cap Q, & \exists y \in (\hat C\cap Q)_1 \colon d_1^0y = d_1^0\cdots d_1^0d_i^1x,\\ \hat C\cap Q\setminus \star_{P}(d_1^0\cdots d_1^0d_i^1x), & \mbox{else.} \end{array}\right.$$ 
One easily checks that $d_i^1x \in \hat C\Leftrightarrow x\in \hat C$. Therefore the inclusion $|\hat C\cap Q| \hookrightarrow |\hat C|$ is a homotopy equivalence. It follows that the inclusion $|C'| \hookrightarrow |\hat C|$ is a homotopy equivalence if there exists an element $y \in (\hat C\cap Q)_1$ such that $d_1^0y = d_1^0\cdots d_1^0d_i^1x$. Suppose that there is no such element. Then our hypotheses imply that  ${\hat C\cap Q\cap\star_{P}(d_1^0\cdots d_1^0d_i^1x)} \subseteq \{d_1^0\cdots d_1^0d_i^1x, e^1_ix\}$. One easily checks that 
$d_1^0\cdots d_1^0d_i^1x \in \hat C\cap Q \Leftrightarrow e^1_ix \in \hat C\cap Q$. It follows that the inclusion $|C'| \hookrightarrow |\hat C\cap Q|$ is a homotopy equivalence. Therefore the inclusion $|C'| \hookrightarrow |\hat C|$ is a homotopy equivalence in this case, too. Hence $C$ is deformable into $C'$.
\end{sloppypar}
It remains to check the accessibility condition of Theorem \ref{di1lem6}. Note first that from every vertex in $\hat C$, there exists a path in $Q$ to a vertex of $C$. Consider a vertex $c' \in C'$. Suppose first that $c'\not= d^0_1 \cdots d^0_1d^1_ix$. Since $c' \in \hat C$, there exists a vertex $c \in C$ such that $c'\to_Q c$. Suppose that  $c \to_P d^0_1 \cdots d^0_1d^1_ix$. Then $c' \to_P d^0_1 \cdots d^0_1d^1_ix$. By considering a path of minimal length between these vertices, we obtain $c' \to_Q d^0_1 \cdots d^0_1d^1_ix$. Since the only edge ending in $d^0_1 \cdots d^0_1d^1_ix$ is $e^1_ix$, $c' \to_Q d^0_1 \cdots d^0_1x$.  Suppose now that $c' = d^0_1 \cdots d^0_1d^1_ix$. Then there exists an edge $y \in (\hat C\cap Q)_1$ such that $d_1^0y = d_1^0\cdots d_1^0d_i^1x$. Then $d^1_1y \in \hat C$ and there exists a vertex $c\in C$ such that $d^1_1y\to_Q c$. Thus $c'\to_Q c$. Suppose that $c\to_P d_1^0\cdots d_1^0d_i^1x$. Then $d^1_1y \to_P d_1^0\cdots d_1^0d_i^1x$. Since $d^0_1y \not= d^0_1e^1_ix$, we have $y \not= e^1_ix$ and thus $d^1_1y \not= d_1^0\cdots d_1^0d_i^1x$. As before, it follows that $d^1_1y \to_Q d_1^0\cdots d_1^0x$. Hence $c' \to_Q d_1^0\cdots d_1^0x$. 
\end{proof}

\subsection{Another concrete example} \label{anotherex}
\begin{figure}
\begin{tikzpicture}[initial text={},on grid] 
\path[pattern=north west lines, pattern color=lightgray]
(0,-1)--(-0.5,-1.5)--(1.5,-1.5)--(1,-1)--cycle;

\path[pattern=north west lines, pattern color=lightgray]
(0,1)--(-0.5,0.75)--(-0.5,1.5)--(0.25,1.5)--cycle;

\path[pattern=north west lines, pattern color=lightgray]
(2,1)--(2.5,1.5)--(2.5,-0.5)--(2,0)--cycle;

\draw[dashed] (0,-1) to  (-0.5,-1.5);

\draw[dashed] (1.5,-1.5) to  (1,-1);

\draw[dashed] (0,1) to (-0.5,0.75);

\draw[dashed] (0.25,1.5) to  (0,1);

\draw[dashed] (2,1) to  (2.5,1.5);

\draw[dashed] (2.5,-0.5) to  (2,0);

\path[draw, fill=lightgray] (0,1)--(2,1)--(2,0)--(1,0)--(1,-1)--(0,-1)--cycle;

 \node[state,minimum size=0pt,inner sep =2pt,fill=white,label={[label distance = 0.3cm]315:\scalebox{1}{$z$}}] (p_0) at (0,0)  {}; 
    
\node[state,minimum size=0pt,inner sep =2pt,fill=white,label={[label distance = 0.3cm]315:\scalebox{1}{$y$}}] (q_0) [above=of p_0,xshift=0cm] {}; 	   
    
   \node[state,minimum size=0pt,inner sep =2pt,fill=white] (p_2) [right=of p_0,xshift=0cm] {};
   
   \node[state,minimum size=0pt,inner sep =2pt,fill=white,label={[label distance = 0.3cm]315:\scalebox{1}{$x$}}] (p_1) [above=of p_2,xshift=0cm] {}; 
   
   \node[state,minimum size=0pt,inner sep =2pt,fill=white] (q_1) [right=of p_1,xshift=0cm] {};
   
   \node[state,minimum size=0pt,inner sep =2pt,fill=white] (p_4) [right=of p_2,xshift=0cm] {};
   
   \node[state,minimum size=0pt,inner sep =2pt,fill=white] [below=of p_0, yshift=0cm] (p_3)   {};

   \node[state,minimum size=0pt,inner sep =2pt,fill=white] (p_5) [right=of p_3,xshift=0cm] {};

   \path[->] 
   	(q_0) edge[above] node {} (p_1)
   	(q_0) edge[above] node {} (p_0)
    (p_1) edge[above] node {} (p_2)
    (p_1) edge[above] node {} (q_1)
    (q_1) edge[above] node {} (p_4)
    (p_0) edge[above] node {} (p_2)
    (p_3) edge[below]  node {} (p_5)
    (p_0) edge[left]  node {} (p_3)
    (p_2) edge[right]  node {} (p_5)
    (p_2) edge[right]  node {} (p_4);
   
\end{tikzpicture}
\caption{A precubical set satisfying the conditions of Theorem \ref{di1lem6}}\label{squares}
\end{figure}
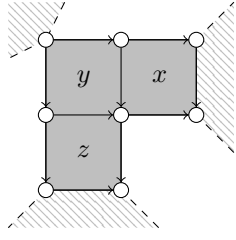

Consider the precubical set $P$ depicted in Figure \ref{squares}. The vertical edges of the squares are the boundaries with lower index $1$, and the horizontal edges are the boundaries with lower index $2$. The dashed areas indicate possible supplementary parts of $P$. We show that the conditions of  Theorem \ref{di1lem6} are satisfied for the $2$-cube $x$ and the free face $d^1_2x$, which is the lower horizontal boundary of $x$.  Let $C$ be any precubical subset of $P$. Define the precubical subset $\hat C\subseteq P$ by
$$\hat C = \left\{\begin{array}{ll} C \cup y_{\sharp}(\llbracket 0, 1\rrbracket^{\otimes 2}) \cup x_{\sharp}(\llbracket 0, 1\rrbracket^{\otimes 2}), & d_2^1x \in C,\\
C \cup y_{\sharp}(\llbracket 0, 1\rrbracket^{\otimes 2}), & d_2^1x \notin C,\; d_1^0d_2^1x, \in C,\\ 
C \cup (d^0_1y)_{\sharp}(\llbracket 0, 1\rrbracket), & d_1^0d_2^1x \notin C,\; d_1^0d_2^1y, \in C,\\
C, & \mbox{else}.
 \end{array}\right.$$ 
We note that $C \subseteq \hat C$ and that from every vertex in $\hat C$ there exists a path in $Q$ to a vertex of $C$. Consider the edges $d^0_1z$ and $d^1_1z$, which are the left and right vertical boundaries of $z$, and define the precubical subset $C'\subseteq Q$ by
$$C' = \left\{\begin{array}{ll} \hat C\cap Q, & d_1^1z \in C,\\
\hat C\cap Q \setminus \star_P (d^0_1d^1_2x), & d_1^1z \notin C,\; d_1^0z \in C,\\
(\hat C\cap Q \setminus \star_P (d^0_1d^1_2x))\setminus \star_P (d^0_1d^1_2y), & d_1^1z \notin C,\; d_1^0z \notin C.
 \end{array}\right.$$ 
Then $C'\subseteq \hat C$, and the inclusion $|C'|\hookrightarrow |\hat C|$ is a homotopy equivalence. Indeed, since $d^1_2x \in \hat C \Leftrightarrow x \in \hat C$, $d^1_2y \in \hat C \Leftrightarrow y \in \hat C$, $d^1_1d^1_1y \in \hat C \Leftrightarrow d^1_1y \in \hat C$, and $d^1_1d^0_1y \in \hat C \Leftrightarrow d^0_1y \in \hat C$, the inclusion $|C'| \hookrightarrow |\hat C|$ decomposes as a sequence of identities and elementary cube collapses:
$$C' = \left\{\begin{array}{ll} \hat C\setminus \{x,d^1_2x\}, & d_1^1z \in C,\\
((\hat C\setminus \{x,d^1_2x\})\setminus \{y,d^1_2y\})\setminus \{d^1_1y,d^1_1d^1_1y\}, & d_1^1z \notin C,\\ & \quad d_1^0z \in C,\\
(((\hat C\setminus \{x,d^1_2x\})\setminus \{y,d^1_2y\})\setminus \{d^1_1y,d^1_1d^1_1y\})\setminus \{d^0_1y,d^1_1d^0_1y\}, & \mbox{else}.
 \end{array}\right.$$ 
Let $c'$ be a vertex of $C'$. We have to show that there exists a vertex $c$ in $C$ such that $c' \to_Q c$ and $c\to_P d^0_1d^1_2x \Rightarrow c'\to_Q d^0_1d^0_1x$. Suppose that this is not the case. Since there exists a path in $Q$ from $c'$ to a vertex of $C$, we have $c'\to_P d^0_1d^1_2x$ but  $c' \not \to_Q d^0_1d^0_1x$.  Since for every vertex $w \in P \setminus y_{\sharp}(\llbracket 0,1 \rrbracket^{\otimes 2})$, $w \to_P d^0_1d^1_2x$  if and only if $w \to_Q d^0_1d^0_1x$, we have $c'\in y_{\sharp}(\llbracket 0,1 \rrbracket^{\otimes 2})$ and thus either $c' = d^1_1d^1_1y$ or $c' = d^1_1d^0_1y$. If $c' = d^1_1d^1_1y = d^0_1d^1_2x$, then $d^1_1z \in C$ and $(d^1_1z)_{\sharp}$ is a path in $Q$ from $c'$ to a vertex $c$ in $C$ such that $c\to_P d^0_1d^1_2x \Rightarrow c'\to_Q d^0_1d^0_1x$.  But we suppose that such a path does not exist. If $c' = d^1_1d^0_1y = d^0_1d^1_2y$, then $d^1_1z \in C$ or $d^0_1z \in C$ and again we obtain a path in $Q$ from $c'$ to a vertex $c$ in $C$ such that $c\to_P d^0_1d^1_2x \Rightarrow c'\to_Q d^0_1d^0_1x$. This is a contradiction. It follows that the conditions of Theorem \ref{di1lem6} are satisfied in the example and hence that the inclusion $|P\setminus \star(d^1_2x)| \hookrightarrow |P|$ induces an isomorphism of homology graphs. 

\subsection{Vertex star collapses} 

Let $x$ be a regular cube of degree $n\geq 2$ of a precubical set $P$, and let $k_1, \dots, k_n \in\{0,1\}$ such that at least one $k_i = 0$, at least one $k_i = 1$, and $\star(d^{k_n}_1\cdots d^{k_1}_1x) \subseteq x_{\sharp}(\llbracket 0,1 \rrbracket ^{\otimes n})$. Consider the precubical subset $Q = {P \setminus \star(d^{k_n}_1\cdots d^{k_1}_1x)}$ of $P$.

\begin{lem} \label{homotopylemma}
Let $C \subseteq P$ be a precubical subset such that one of the following two conditions holds:
\begin{enumerate}[(i)]
\item $\forall\, l_1, \dots, l_n \in \{0,1\} : d^{l_n}_1\cdots d^{l_1}_1x \in C \Rightarrow x_{\sharp}(\llbracket 0,l_1\rrbracket \otimes \cdots \otimes \llbracket 0,l_n\rrbracket ) \subseteq C$.
\item $\forall\, l_1, \dots, l_n \in \{0,1\} : d^{l_n}_1\cdots d^{l_1}_1x \in C \Rightarrow x_{\sharp}(\llbracket l_1,1 \rrbracket \otimes \cdots \otimes \llbracket l_n,1\rrbracket ) \subseteq C.$
\end{enumerate}
Then the inclusion $|C\cap Q| \hookrightarrow |C|$ is a homotopy equivalence.
\end{lem}

\begin{proof}
We prove the lemma under hypothesis (i). The proof under hypothesis (ii) is analogous and is left to the reader. If $d^{k_n}_1\cdots d^{k_1}_1x \notin C$, the lemma holds because then $C\cap  \star_P(d^{k_n}_1\cdots d^{k_1}_1x) = \emptyset$. Suppose that $d^{k_n}_1\cdots d^{k_1}_1x \in C$. Write $Y = \star_{\llbracket 0,1\rrbracket^{\otimes n}}((k_1, \dots, k_n)) = \{y \in \llbracket 0,1\rrbracket ^{\otimes n} :  y_j \in \{k_j,[0,1]\}\}$. We have $x_{\sharp}(Y) = \star_P(d^{k_n}_1\cdots d^{k_1}_1x)$. Fix an index $i$ such that $k_i = 1$. For $y \in Y$, define $\hat y \in Y$ by $$\hat y_j = \left\{\begin{array}{ll} y_j, & j\not= i, \\ 1, & j = i.\end{array}\right.$$ Let $Z$ be a maximal subset of $Y$ such that 
\begin{enumerate}[(i)]
\item for all $y \in Y$, $y \in Z \Leftrightarrow \hat y \in Z$;
\item $C \setminus x_{\sharp}(Z)$ is a precubical subset of $C$; and
\item the inclusion $|C\setminus x_{\sharp}(Z)| \hookrightarrow |C|$ is a homotopy equivalence.
\end{enumerate} 
Note that such a maximal subset of $Y$ exists because the empty set satisfies these conditions. Therefore the result follows if we can show that $Z = Y$.  Suppose that this is not the case. Let $y \in Y\setminus Z$ be an element of maximal degree. By condition (i), $\hat y \in Y \setminus Z$ and $y_i = [0,1]$. Let $r$ be the number of indexes $j\in \{1, \dots, i\}$ such that $y_j = [0,1]$. Then $\hat y = d^1_ry$. 

Suppose for a moment that $d^1_rx_{\sharp}(y) = x_{\sharp}(\hat y)\notin C \setminus x_{\sharp}(Z)$. Then by condition (ii), $x_{\sharp}(y) \notin C \setminus x_{\sharp}(Z)$. Thus $C \setminus x_{\sharp} (Z \cup \{y, \hat y\}) = C \setminus x_{\sharp}(Z)$. Hence $Z \cup \{y, \hat y\}$ satisfies conditions (i)-(iii). Since $Z$ is maximal, this is impossible. It follows that $d^1_rx_{\sharp}(y) \in C \setminus x_{\sharp}(Z)$.

\begin{sloppypar}
We show that also $x_{\sharp}(y) \in C \setminus x_{\sharp}(Z)$. Set $$l_j = \left\{\begin{array}{ll} 1, & y_j = [0,1], \\ y_j, & \mbox{else.} \end{array}\right.$$ Then $(l_1, \dots, l_n) = d^1_1\cdots d^1_1y = d^1_1\cdots d^1_1 d^1_ry$ and therefore  $d^{l_n}_1\cdots d^{l_1}_1x = d^1_1\cdots d^1_1 d^1_rx_{\sharp}(y) \in C$. Moreover,  $x_{\sharp}(y) \in x_{\sharp}(\llbracket 0,l_1\rrbracket \otimes \cdots \otimes \llbracket 0,l_n\rrbracket )$. Hence $x_{\sharp}(y) \in C$, and thus $x_{\sharp}(y) \in C \setminus x_{\sharp}(Z)$.
\end{sloppypar}

We show that $\star_{C \setminus x_{\sharp}(Z)}(d^1_rx_{\sharp}(y)) = \{d^1_rx_{\sharp}(y),x_{\sharp}(y)\}$. Suppose that this is not the case. Suppose first that there exists an element $c \in (C\setminus x_{\sharp}(Z))\setminus\{x_{\sharp}(y)\}$ that has $d^1_rx_{\sharp}(y) = x_{\sharp}(\hat y)$ in its boundary. Then $c \in x_{\sharp}(Y)$ and there exists an element $y' \in Y\setminus Z$ such that $y' \not= y$ and $x_{\sharp}(y') = c$. Since $y' \not= y$, $y'_i = 1$ and there exists an element $y'' \in Y$ such that $\hat y'' = y'$. By condition (i) above, $y'' \in Y \setminus Z$. Since $y$ is of maximal degree, this is impossible. Suppose now that there exists an element $c \in C\setminus x_{\sharp}(Z)$ that has $x_{\sharp}(y)$ in its boundary. Then $c \in  x_{\sharp}(Y)$ and there exists an element $y' \in Y\setminus Z$ such that $x_{\sharp}(y') = c$. Since $y$ is of maximal degree, this is impossible. It follows that $\star_{C \setminus x_{\sharp}(Z)}(d^1_rx_{\sharp}(y)) = \{d^1_rx_{\sharp}(y),x_{\sharp}(y)\}$.

Therefore $C \setminus x_{\sharp}(Z \cup \{y,\hat y\}) = (C \setminus x_{\sharp}(Z)) \setminus \star_{C \setminus x_{\sharp}(Z)}(d^1_rx_{\sharp}(y))$ is a precubical subset of $C$. Since $x$ is regular, so is $x_{\sharp}(y)$. Hence the inclusion $$|(C \setminus x_{\sharp}(Z)) \setminus \star_{C \setminus x_{\sharp}(Z)}(d^1_rx_{\sharp}(y))|\hookrightarrow |(C \setminus x_{\sharp}(Z))|$$ is a homotopy equivalence. It follows that 
the set $Z \cup \{y,\hat y\}$ satisfies conditions (i)-(iii) above. Since $Z$ is maximal, this is impossible. Thus $Z = Y$.
\end{proof}

\begin{sloppypar}
\begin{theor} \label{HV}
The inclusion $\iota \colon |Q| \hookrightarrow |P|$ induces a graph isomorphism ${\iota_*\colon H_*(|Q|) \to H_*(|P|)}$.
\end{theor}
\end{sloppypar}

\begin{proof}
By Lemma \ref{homotopylemma}, $\iota$ is a homotopy equivalence. Therefore it induces an isomorphism in homology. 
Let $\alpha$ and $\beta$ be two homology classes in $H_*(|Q|)$ such that $\iota_*(\alpha) \nearrow \iota_*(\beta)$. 
By Theorem \ref{morpoint}, it suffices to show that $\alpha \nearrow \beta$. Let $A, B \subseteq P$ be precubical subsets such that $\iota_*(\alpha) \in \im\, H_*(|A| \hookrightarrow |P|)$, $\iota_*(\beta) \in \im\, H_*(|B| \hookrightarrow |P|)$, and for all $a \in A_0$ and $b \in B_0$, $a\to_P b$. Consider the precubical subsets $\hat A, \hat B \subseteq P$ given by $$\hat A = A \cup \bigcup \limits_{d^{m_n}_1\cdots d^{m_1}_1x \in A }  x_{\sharp}(\llbracket 0,m_1\rrbracket \otimes \cdots \otimes \llbracket 0,m_n\rrbracket )$$ and $$\hat B = B \cup \bigcup \limits_{d^{m_n}_1\cdots d^{m_1}_1x \in B }  x_{\sharp}(\llbracket m_1,1\rrbracket \otimes \cdots \otimes \llbracket m_n,1\rrbracket ).$$
Then $\hat A$ and $\hat B$ satisfy conditions (i) and (ii) of Lemma \ref{homotopylemma} respectively. Set $A' = \hat A \cap Q$ and $B' = \hat B \cap Q$. By Lemma \ref{homotopylemma}, the inclusions $|A'| \hookrightarrow |\hat A|$ and $|B'| \hookrightarrow  |\hat B|$ are homotopy equivalences. 
By Lemma \ref{deflemma}, $\alpha \in \im\, H_*(|A'| \hookrightarrow |Q|)$ and  $\beta \in \im\, H_*(|B'| \hookrightarrow |Q|)$. Consider vertices $a' \in  A'_0$ and $b' \in B'_0$. By definition of $\hat A$ and $\hat B$, there exist vertices $a \in A_0$ and $b\in B_0$ such that $a' \to_P a$ and $b \to _P  b'$. Since $\iota_*(\alpha) \nearrow \iota_*(\beta)$, we have $a \to_P b$ and hence $a' \to_P b'$. By Theorem  \ref{omegaprime}, $a' \to_Q b'$. This shows that $\alpha \nearrow \beta$.
\end{proof}

\section{Topological abstraction} \label{topabs}

We are now ready to define the preorder relation of topological abstraction. We relate it to homeomorphic abstraction and discuss topological abstraction through cube collapses. As before, we consider homology with coefficients in an arbitrary commutative unital ring.

\subsection{The preorder relation of topological abstraction} \label{topabdef}
Consider two $M\textrm{-}$HDAs $\A = (P,I,F,\lambda)$ and $\B = (Q,J,G,\mu)$. We write $\B \stackrel{\sim}{\to} \A$ and say that $\B$ is a \emph{topological abstraction} of $\A$, or that $\A$ is a \emph{topological refinement} of $\B$, if there exists a weak morphism $f$ from $\B$ to $\A$ such that 
\begin{enumerate}
\item $f(J) = I$, $f(G) = F$, $f(m_0(Q)) = m_0(P)$, $f(m_1(Q)) = m_1(P)$; 
\item $f$ is a homotopy equivalence;
\item the functor $f_*\colon TC(\B) \to TC(\A)$ is an isomorphism; and
\item the map $f_* \colon H_*(|Q|) \to H_*(|P|)$ is a graph isomorphism.
\end{enumerate}
It is clear that  $\stackrel{\sim}{\to}$ is a preorder relation on the class of $M$-HDAs.

\subsection{Homeomorphic abstraction \cite{weakmor}}  Consider two $M\textrm{-}$HDAs $\A = (P,I,F,\lambda)$ and $\B = (Q,J,G,\mu)$. We say that $\B$ is a \emph{homeomorphic abstraction} of $\A$, or that $\A$ is a \emph{homeomorphic refinement} of $\B$, if there exists a weak morphism $f$ from $\B$ to $\A$ that is a homeo\-morphism and satisfies $f(J) = I$ and $f(G) = F$. Such a weak morphism can be considered a label-preserving $\mbox{T-}$homotopy equivalence in the sense of Gaucher and Goubault \cite{GaucherT, GaucherGoubault}. We use the notation $\B \stackrel{\approx}{\to} \A$ to indicate that $\B$ is a homeomorphic abstraction of $\A$. The relation $\stackrel{\approx}{\to}$ is a preorder on the class of $M$-HDAs.

\begin{theor} \label{homeoabs}
Suppose that $\B \stackrel{\approx}{\to} \A$ and that $\A$ is weakly regular. Then $\B\stackrel{\sim}{\to} \A$. Moreover, $\A$ is (co)accessible if and only if $\B$ is (co)accessible.
\end{theor}

\begin{proof}
The first statement follows from Proposition \ref{weakreg}, \cite[4.8.1, 4.9.1]{weakmor}, and \cite[5.8]{hgraph}. Let $f$ be a weak morphism from $\B$ to $\A$ that is a homeo\-morphism and satisfies $f(J) = I$ and $f(G) = F$. Suppose that $\A$ is accessible. Consider a vertex $b \in Q_0$. Then there exists a path in $P$ from an initial state $w \in I$ to $f_0(b)$. Since $f_0(J) = I$, there exists a vertex $v\in J$ such that $f_0(v) = w$. By \cite[Section 4.3, Remark (i)]{weakmor}, there exists a path in $Q$ from $v$ to $b$. It follows that $\B$ is accessible. In the same way one shows that $\B$ is coaccessible if $\A$ is coaccessible. 

In order to show that $\A$ is (co)accessible if $\B$ is (co)accessible, we establish the condition of Proposition \ref{coacc}. Let $a$ be a vertex of $P$. Consider the carrier of $a$, i.e., the unique element $c(a) \in Q$ for which there exists an element $u \in \mathopen] 0,1\mathclose[^{\deg(c(a))}$ such that $f([c(a),u]) = [a,()]$ (cf. \cite[4.3.1]{weakmor}). If $\deg(c(a)) = 0$, then $f_0(c(a)) = a$ and thus $f_0(c(a)) \to_P a \to_P f_0(c(a))$. Suppose that $\deg(c(a)) = n > 0$. Since $f\circ |c(a)_{\sharp}| = |c(a)_{\flat}|\circ \phi_{c(a)}$, we have $[a,()] \in |c(a)_{\flat}|(|R_{c(a)}|)$. It follows that there exists a vertex $y \in R_{c(a)}$ such that $c(a)_{\flat}(y) = a$. Write $R_{c(a)} = \llbracket 0, k_1\rrbracket \otimes \cdots \otimes \llbracket 0, k_n\rrbracket $. Since $(0,\dots, 0) \to_{R_{c(a)}} y \to_{R_{c(a)}} (k_1, \dots , k_n)$, we have $c(a)_{\flat}(0,\dots, 0) \to_P a \to_P c(a)_{\flat}(k_1,\dots,k_n)$. By \cite[2.1.1]{weakmor}, $(\phi_{c(a)})_0(0,\dots , 0) = (0,\dots ,0)$ and $(\phi_{c(a)})_0(1,\dots ,1) = (k_1,\dots ,k_n)$. By \cite[2.2.1]{weakmor}, it follows that $f_0(c(a)_{\sharp}(0,\dots 0)) \to_P a \to_P f_0(c(a)_{\sharp}(1,\dots 1))$. 
\end{proof}

\subsection{Elementary collapses} \label{di1red}

Consider a weakly regular $M$-HDA $\A = (P,I, F,\lambda)$ and a regular element $x$ of degree $n \geq 2$ with free face $d^ k_ix$. Suppose that there exists no element $y \in P_1\setminus \{e^k_ix\}$ such that $d^k_1y = d_1^{1-k} \cdots d_1^{1-k}d_i^kx$. If $n\leq 3$, suppose furthermore that
$d_1^{1-k} \cdots d_1^{1-k}d_i^kx \notin I\cup F$. If $n = 2$, suppose finally that there exists an edge $y \in P_1\setminus \{d^k_ix\}$ such that $d_1^{1-k}y = d_1^{1-k}d_i^kx$. Consider the precubical subset $Q = P \setminus \star(d_i^kx)$ of $P$ and the $M$-HDA $\B = (Q,I,F,\lambda|_{Q_1})$.

\begin{theor} \label{freefacecollapse}
$\B \we \A$. Moreover, $\A$ is (co)accessible if and only if $\B$ is (co)accessible.
\end{theor}

\begin{proof}
Once again, we only prove the theorem for $k=1$. The inclusion $\iota\colon |Q| \hookrightarrow |P|$ is a weak morphism of $M$-HDAs from $\B$ to $\A$ that is a homotopy equivalence. By Theorem \ref{red1}, $\iota_*\colon H_*(|Q|) \to H_*(|P|)$ is a graph isomorphism. If $n\geq 3$, then $Q_{\leq 1} = P_{\leq 1}$ and thus $m_j(Q) = m_j(P)$ $(j = 0,1)$. For $n=2$, this holds by Lemma \ref{minmax}. It follows that $\iota$ is a bijection on the sets of initial, final, maximal, and minimal vertices. If $n\geq 4$, then $Q_{\leq 2} = P_{\leq 2}$ and therefore the functor $\iota_*\colon TC(\B) \to TC(\A)$ is an isomorphism. As observed in Example \ref{excancel}, in the case $n = 2$, this follows from Theorem \ref{TC2}. Suppose that $n=3$. Consider a path $\omega \in P^{\mathbb I}$ from a vertex $v \in I \cup F \cup m_0(P) \cup m_1(P)$ to $d^0_1d^0_1d^1_ix$. By our hypotheses, $v \not= d^0_1d^0_1d^1_ix$. It follows that $\omega$ ends with the edge  $e^1_ix$. By Theorem \ref{TC3}, $\iota_*\colon TC(\B) \to TC(\A)$ is an isomorphism.

If $n\geq 3$, we have $Q_{\leq 1} = P_{\leq 1}$, and this implies that $\A$ is (co)accessible if and only if $\B$ is (co)accessible. Suppose that $n = 2$. Since $\iota_0\colon Q_0 \to P_0$ is the identity, Proposition \ref{coacc} implies that $\A$ is (co)accessible if $\B$ is (co)accessible. Suppose that $\A$ is accessible, and consider a vertex $b \in Q_0$. Then there exists a path $\omega \in P^{\mathbb I}$ from a vertex in $I$ to $b$. As explained in Example \ref{excancel}, we may apply Lemma \ref{aprime} to obtain a path $\omega ' \in Q^{\mathbb I}$ such that $\omega ' \sim \omega$. It follows that $\B$ is accessible. Suppose that $\A$ is coaccessible but  $\B$ is not. Then there exists a vertex $b \in Q_0$ such that there exists a path $\omega \in P^{\mathbb I}$ from $b$ to a final state but there is no path in $Q$ from $b$ to a final state. Since $\omega \notin Q^{\mathbb I}$, it begins with a path from $b$ to $d^0_1d^1_ix$, and a shortest such path is a path in $Q$. Let $y \in Q_1 = P_1\setminus\{d^1_ix\}$ be an edge such that $d^0_1y = d^0_1d^1_ix$. Then $y \not= e^ 1_ix$ and therefore $d^1_1y \not= d^0_1d^1_ix$. Since there is no path in $Q$ from $b$ to a final state but a path in $Q$ from $b$ to $d^1_1y$, there is no path in $Q$ from $d^1_1y$ to a final state. Since $\A$ is coaccessible, there exists a path $\nu \in P^{\mathbb I}$ from $d^1_1y$ to a vertex $c \in F$. Since $\nu \notin Q^{\mathbb I}$ and the only edge ending in $d^0_1d^1_ix$ is $e^1_ix = d^0_{3-i}x$, $\nu$ begins with a path from $d^1_1y$ to $d^0_1d^0_1x$, and if  $\alpha$ is a shortest such path, then $\alpha\in Q^{\mathbb I}$. On the other hand, $\nu$ terminates with a path from $d^1_1d^1_1x$ to $c$. Let $\beta $ be a shortest such path. Then $\beta\in Q^{\mathbb I}$. The concatenation $\alpha \cdot (d^0_ix)_{\sharp}\cdot (d^1_{3-i}x)_{\sharp}\cdot \beta$ is a path in $Q$ from $d^1_1y$ to $c$. This is a contradiction. Hence $\B$ is coaccessible if $\A$ is coaccessible.
\end{proof}

\subsection{Vertex star collapses} 

Let $x$ be a regular cube of degree $n \geq 2$ of an $M$-HDA $\A = (P,I, F,\lambda)$, and let $k_1, \dots, k_n \in\{0,1\}$ be numbers such that at least one $k_i = 0$, at least one $k_i = 1$, $d^{k_n}_1\cdots d^{k_1}_1x \notin I \cup F$, and $\star(d^{k_n}_1\cdots d^{k_1}_1x) \subseteq x_{\sharp}(\llbracket 0,1 \rrbracket ^{\otimes n})$. 
Consider the precubical subset $Q = P\setminus \star(d^{k_n}_1\cdots d^{k_1}_1x)$ of $P$ and the $M\mbox{-}$HDA $\B = (Q,I,F,\lambda|_{Q_1})$.

\begin{theor} \label{red2}
$\B \we \A$. Moreover, $\A$ is (co)accessible if and only if $\B$ is (co)accessible. 
\end{theor}

\begin{proof}
By Lemma \ref{homotopylemma}, the inclusion ${\iota \colon |Q| \hookrightarrow |P|}$ is a a weak morphism from $\B$ to $\A$ that is a homotopy equivalence. By Lemma \ref{mj}, $\iota$ is a bijection between the sets of initial, final, maximal, and minimal vertices. By Theorem \ref{TCV}, the functor $\iota_*\colon TC(\B) \to TC(\A)$ is an isomorphism. By Theorem \ref{HV}, $\iota_*\colon H_*(|Q|) \to H_*(|P|)$ is a graph isomorphism. It follows that $\B \we \A$. 

Since $d^0_1\cdots d^0_1x \to_P d^{k_n}_1\cdots d^{k_1}_1x \to_P d^1_1\cdots d^1_1x$, Proposition \ref{coacc} implies that $\A$ is (co)accessible if $\B$ is (co)accessible. Suppose that $\A$ is accessible, and let $b$ be a vertex in $Q$. Then there exists a path $\alpha\in P^{\mathbb I}$ from an initial state $v$ to $b$. By Theorem  \ref{omegaprime}, there exists a path in $Q$ to $v$ to $b$. Thus $\B$ is accessible. An analogous argument shows that $\B$ is coaccessible if $\A$ is coaccessible.
\end{proof}

\section{An example: Peterson's mutual exclusion algorithm} \label{example}

In this section, we use our results to construct a small HDA model for Peterson's mutual exclusion algorithm. The principal properties of Peterson's algorithm can be derived from this model and the local independence relation associated with the system.

\subsection{Peterson's algorithm}
Peterson's algorithm is a protocol designed to give two processes fair and mutually exclusive access to a shared resource \cite{Peterson}. The algorithm is based on three shared variables, namely the turn variable $t$ whose possible values are the process IDs---let us assume that these are $0$ and $1$---and  the boolean variables $b_0$ and $b_1$, both initially $0$. Process $i$ executes the protocol given by the program graph depicted in Figure \ref{peterPG}. 
\begin{figure}[] 
\begin{tikzpicture}[initial text={},on grid] 
     
 \node[state,accepting, initial by arrow, initial where=left, initial distance=0.2cm,minimum size=0pt,inner sep =1pt,fill=white] (q_0)   {\scalebox{0.85}{$\ell_0$}}; 
 
 \node[state,minimum size=0pt,inner sep =1pt,fill=white] (q_1) [below=of q_0,yshift=-1cm] {\scalebox{0.85}{$\ell_1$}};
    
 \node[state,minimum size=0pt,inner sep =1pt,fill=white] (q_2) [right=of q_1,xshift=1cm] {\scalebox{0.85}{$\ell_2$}};
   
 \node[state,minimum size=0pt,inner sep =1pt,fill=white] (q_3) [above=of q_2,yshift=1cm] {\scalebox{0.85}{$\ell_3$}};

    \path[->] 
    (q_0) edge[left,dotted,inner sep =5pt,out=260,in=80] node {\scalebox{0.85}{$b_i\!\coloneqq \!1$}} (q_1)
    (q_1) edge[above,dotted,inner sep =5pt,out=350,in=170] node {\scalebox{0.85}{$t\!\coloneqq \!1-i$}} (q_2)
    (q_2) edge[right,dotted,out=80,in=260] node {\scalebox{0.85}{$(b_{1-i}\!= \!0 \vee t\! =\! i)\colon\!crit$}} (q_3)
    (q_3) edge[above,dotted,out=170,in=350] node {\scalebox{0.85}{$b_i\!\coloneqq \!0$}} (q_0);

\end{tikzpicture} 

\caption{Program graph for process $i$ in Peterson's algorithm}\label{peterPG}
\end{figure}
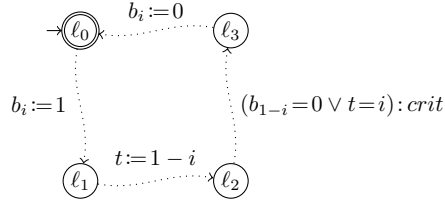
It first sets variable $b_i$ to $1$ in order to indicate that it intends to enter the critical section and access the shared resource. Then it gives priority to the other process by setting the turn variable $t$ to $1-i$. After this, it waits until the other process does not intend to enter the critical section or it is its own turn to do so. When one of these conditions is satisfied, it enters the critical section and accesses the shared resource. The corresponding action $crit$ has no effect on the shared variables $t, b_0,$ and $b_1$. After having used the shared resource, process $i$ leaves the critical section and sets variable $b_i$ to $0$ again. The procedure is now repeated arbitrarily often or forever.

\subsection{HDA models} An HDA model $\A$ of the accessible part of the system given by Peterson's algorithm is depicted in Figure \ref{peterHDA}. In the upper initial and final state, $t = 0$; in the lower one, $t = 1$. The monoid of labels is the free monoid $M = \Sigma ^*$ where $\Sigma $ consists of the actions of the processes indexed by the process IDs: $$\Sigma = \{b_i:=_i1, \; b_i:=_i0,\; t:=_i1-i,\;crit_i\; :\; i=0,1\}.$$

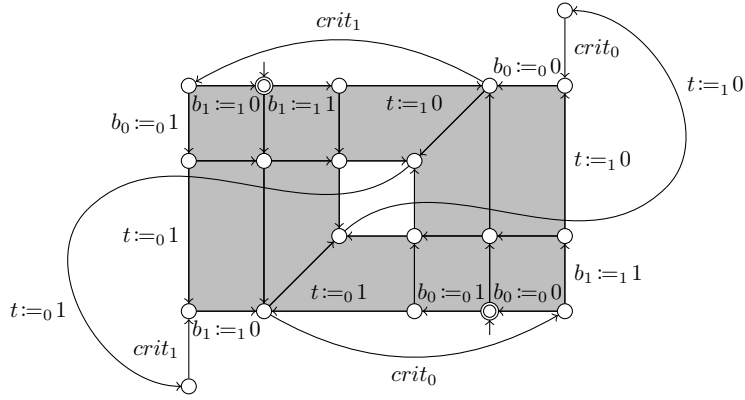
\begin{figure}[]
\begin{tikzpicture}[initial text={},on grid]

\path[draw, fill=lightgray] (0,0)--(1,0)--(1,-1)--(0,-1)--cycle;    
  
\path[draw, fill=lightgray] (1,0)--(2,0)--(2,-1)--(1,-1)--cycle;

\path[draw, fill=lightgray] (2,0)--(4,0)--(3,-1)--(2,-1)--cycle;

\path[draw, fill=lightgray] (0,-1)--(1,-1)--(1,-3)--(0,-3)--cycle;

\path[draw, fill=lightgray] (1,-1)--(2,-1)--(2,-2)--(1,-3)--cycle;

\path[draw, fill=lightgray] (3,-1)--(4,0)--(4,-2)--(3,-2)--cycle;

\path[draw, fill=lightgray] (4,0)--(5,0)--(5,-2)--(4,-2)--cycle;

\path[draw, fill=lightgray] (2,-2)--(3,-2)--(3,-3)--(1,-3)--cycle;

\path[draw, fill=lightgray] (3,-2)--(4,-2)--(4,-3)--(3,-3)--cycle;

\path[draw, fill=lightgray] (4,-2)--(5,-2)--(5,-3)--(4,-3)--cycle;

\draw[] (3,-1) to [out=225,in=45] (-1.35,-1.65); 

\draw[->] (-1.35,-1.65) to [out=225,in=180] (-0.1,-4);

\node at (-2,-3) {\scalebox{0.85}{$t\!\coloneqq_0\! 1$}};

\draw[] (2,-2) to [out=45,in=225] (6.35,-1.35); 

\draw[->] (6.35,-1.35) to [out=45,in=0] (5.1,1);

\node at (7,0) {\scalebox{0.85}{$t\!\coloneqq_1\! 0$}};

\node[state,minimum size=0pt,inner sep =2pt,fill=white] (q_0) at (0,0)  {}; 
    
\node[state,minimum size=0pt,inner sep =2pt,fill=white,accepting, initial,initial where=above,initial distance=0.2cm] (q_1) [right=of q_0,xshift=0cm] {};
   
\node[state,minimum size=0pt,inner sep =2pt,fill=white] (q_2) [right=of q_1,xshift=0cm] {};

\node[state,minimum size=0pt,inner sep =2pt,fill=white] (q_3) [right=of q_2,xshift=1cm] {};

\node[state,minimum size=0pt,inner sep =2pt,fill=white] (q_4) [right=of q_3,xshift=0cm] {};
   
\node[state,minimum size=0pt,inner sep =2pt,fill=white] (q_5) [below=of q_0,xshift=0cm] {};  

\node[state,minimum size=0pt,inner sep =2pt,fill=white] (q_6) [right=of q_5,xshift=0cm] {};

\node[state,minimum size=0pt,inner sep =2pt,fill=white] (q_7) [right=of q_6,xshift=0cm] {};

\node[state,minimum size=0pt,inner sep =2pt,fill=white] (q_8) [right=of q_7,xshift=0cm] {};

\node[state,minimum size=0pt,inner sep =2pt,fill=white] (q_9) [below=of q_7,xshift=0cm] {};

\node[state,minimum size=0pt,inner sep =2pt,fill=white] (q_10) [right=of q_9,xshift=0cm] {};

\node[state,minimum size=0pt,inner sep =2pt,fill=white] (q_11) [right=of q_10,xshift=0cm] {};

\node[state,minimum size=0pt,inner sep =2pt,fill=white] (q_12) [right=of q_11,xshift=0cm] {};

\node[state,minimum size=0pt,inner sep =2pt,fill=white] (q_17) [below=of q_12,xshift=0cm] {};

\node[state,minimum size=0pt,inner sep =2pt,fill=white,accepting,initial,initial where=below,initial distance=0.2cm] (q_16) [left=of q_17,xshift=0cm] {};

\node[state,minimum size=0pt,inner sep =2pt,fill=white] (q_15) [left=of q_16,xshift=0cm] {};

\node[state,minimum size=0pt,inner sep =2pt,fill=white] (q_14) [left=of q_15,xshift=-1cm] {};

\node[state,minimum size=0pt,inner sep =2pt,fill=white] (q_13) [left=of q_14,xshift=0cm] {};

\node[state,minimum size=0pt,inner sep =2pt,fill=white] (p_0) [above=of q_4,xshift=0cm] {};

\node[state,minimum size=0pt,inner sep =2pt,fill=white] (p_1) [below=of q_13,xshift=0cm] {};

    \path[->] 
    (p_0) edge[right] node {\scalebox{0.85}{$crit_0$}} (q_4)
    
    (q_0) edge[below] node {\scalebox{0.85}{$b_1\!\coloneqq_1\! 0$}} (q_1)
    (q_1) edge[below] node {\scalebox{0.85}{$b_1\!\coloneqq_1\! 1$}} (q_2)
    (q_2) edge[below] node {\scalebox{0.85}{$t\!\coloneqq_1\! 0$}} (q_3)
    (q_4) edge[above] node {\scalebox{0.85}{$b_0\!\coloneqq_0\! 0$}} (q_3)
    
   	(q_0) edge[left] node {\scalebox{0.85}{$b_0\!\coloneqq_0\! 1$}} (q_5)
   	(q_1) edge[above] node {} (q_6)
   	(q_2) edge[above] node {} (q_7)
   	(q_3) edge[above] node {} (q_8)
   	
   	(q_5) edge[above] node {} (q_6)
   	(q_6) edge[above] node {} (q_7)
   	(q_7) edge[above] node {} (q_8)
   	
   	(q_7) edge[above] node {} (q_9)
   	(q_10) edge[above] node {} (q_8)
   	(q_11) edge[above] node {} (q_3)
   	(q_12) edge[right] node {\scalebox{0.85}{$t\!\coloneqq_1\! 0$}} (q_4)
   	
   	(q_10) edge[above] node {} (q_9)
   	(q_11) edge[above] node {} (q_10)
   	(q_12) edge[above] node {} (q_11)
 
 	(q_5) edge[left] node {\scalebox{0.85}{$t\!\coloneqq_0\! 1$}} (q_13)
 	(q_6) edge[above] node {} (q_14)
 	(q_14) edge[above] node {} (q_9)
 	(q_15) edge[above] node {} (q_10)
 	(q_16) edge[above] node {} (q_11)
 	(q_17) edge[right] node {\scalebox{0.85}{$b_1\!\coloneqq_1\! 1$}} (q_12)
 	
 	(q_13) edge[below] node {\scalebox{0.85}{$b_1\!\coloneqq_1\! 0$}} (q_14)
 	(q_15) edge[above] node {\scalebox{0.85}{$t\!\coloneqq_0\! 1$}} (q_14)
 	(q_16) edge[above] node {\scalebox{0.85}{$b_0\!\coloneqq_0\! 1$}} (q_15)
 	(q_17) edge[above] node {\scalebox{0.85}{$b_0\!\coloneqq_0\! 0$}} (q_16)
 	
 	(p_1) edge[left] node {\scalebox{0.85}{$crit_1$}} (q_13)
 	
 	(q_3) edge[above, bend right] node {\scalebox{0.85}{$crit_1$}} (q_0)
 	(q_14) edge[below, bend right] node {\scalebox{0.85}{$crit_0$}} (q_17)
  	;
\end{tikzpicture}
\caption{HDA for Peterson's algorithm}\label{peterHDA}
\end{figure}

By Theorem \ref{freefacecollapse}, we obtain a topological abstraction of $\A$ by collapsing the lower left $2$-cube and its left vertical edge. Theorem \ref{red2} permits us then to collapse the star of the lower left vertex of the upper left square. Similarly, we can eliminate the two rightmost $2$-cubes. Merging edges we obtain the HDA in Figure \ref{peterHDA2}, which, by Theorem \ref{homeoabs}, is a
topological abstraction of the HDA for Peterson's algorithm.

\begin{figure}[]
\begin{tikzpicture}[initial text={},on grid]

\path[draw, fill=lightgray] (1,0)--(2,0)--(2,-1)--(1,-1)--cycle;

\path[draw, fill=lightgray] (2,0)--(4,0)--(3,-1)--(2,-1)--cycle;

\path[draw, fill=lightgray] (1,-1)--(2,-1)--(2,-2)--(1,-3)--cycle;

\path[draw, fill=lightgray] (3,-1)--(4,0)--(4,-2)--(3,-2)--cycle;

\path[draw, fill=lightgray] (2,-2)--(3,-2)--(3,-3)--(1,-3)--cycle;

\path[draw, fill=lightgray] (3,-2)--(4,-2)--(4,-3)--(3,-3)--cycle;

\draw[] (3,-1) to [out=225,in=45] (-0.35,-1.65); 

\draw[->] (-0.35,-1.65) to [out=225,in=180] (0.9,-3);

\node[align=left] at (-1,-2) {\scalebox{0.85}{$t\!\coloneqq_0\! 1;$}\\\scalebox{0.85}{$crit_1;$}\\\scalebox{0.85}{$b_1\!\coloneqq_1\! 0$}};

\draw[] (2,-2) to [out=45,in=225] (5.35,-1.35); 

\draw[->] (5.35,-1.35) to [out=45,in=0] (4.1,0);

\node[align=left] at (6.1,-1) {\scalebox{0.85}{$t\!\coloneqq_1\! 0;$}\\\scalebox{0.85}{$crit_0;$}\\\scalebox{0.85}{$b_0\!\coloneqq_0\! 0$}};

\node[state,minimum size=0pt,inner sep =2pt,fill=white,accepting,initial,initial where=above,initial distance=0.2cm] (q_1) at (1,0) {};
   
\node[state,minimum size=0pt,inner sep =2pt,fill=white] (q_2) [right=of q_1,xshift=0cm] {};

\node[state,minimum size=0pt,inner sep =2pt,fill=white] (q_3) [right=of q_2,xshift=1cm] {};

\node[state,minimum size=0pt,inner sep =2pt,fill=white] (q_6) [right=of q_5,xshift=0cm] {};

\node[state,minimum size=0pt,inner sep =2pt,fill=white] (q_7) [right=of q_6,xshift=0cm] {};

\node[state,minimum size=0pt,inner sep =2pt,fill=white] (q_8) [right=of q_7,xshift=0cm] {};

\node[state,minimum size=0pt,inner sep =2pt,fill=white] (q_9) [below=of q_7,xshift=0cm] {};

\node[state,minimum size=0pt,inner sep =2pt,fill=white] (q_10) [right=of q_9,xshift=0cm] {};

\node[state,minimum size=0pt,inner sep =2pt,fill=white] (q_11) [right=of q_10,xshift=0cm] {};

\node[state,minimum size=0pt,inner sep =2pt,fill=white,accepting,initial,initial where=below,initial distance=0.2cm] (q_16) [left=of q_17,xshift=0cm] {};

\node[state,minimum size=0pt,inner sep =2pt,fill=white] (q_15) [left=of q_16,xshift=0cm] {};

\node[state,minimum size=0pt,inner sep =2pt,fill=white] (q_14) [left=of q_15,xshift=-1cm] {};

    \path[->] 
    
    (q_1) edge[below] node {\scalebox{0.85}{$b_1\!\coloneqq_1\! 1$}} (q_2)
    (q_2) edge[below] node {\scalebox{0.85}{$t\!\coloneqq_1\! 0$}} (q_3)

   	(q_1) edge[left] node {\scalebox{0.85}{$b_0\!\coloneqq_0\! 1$}} (q_6)
   	(q_2) edge[above] node {} (q_7)
   	(q_3) edge[above] node {} (q_8)

   	(q_6) edge[above] node {} (q_7)
   	(q_7) edge[above] node {} (q_8)
   	
   	(q_7) edge[above] node {} (q_9)
   	(q_10) edge[above] node {} (q_8)
   	(q_11) edge[right] node {\scalebox{0.85}{$t\!\coloneqq_1\! 0$}} (q_3)

   	(q_10) edge[above] node {} (q_9)
   	(q_11) edge[above] node {} (q_10)

 	(q_6) edge[left] node {\scalebox{0.85}{$t\!\coloneqq_0\! 1$}} (q_14)
 	(q_14) edge[above] node {} (q_9)
 	(q_15) edge[above] node {} (q_10)
 	(q_16) edge[right] node {\scalebox{0.85}{$b_1\!\coloneqq_1\! 1$}} (q_11)

 	(q_15) edge[above] node {\scalebox{0.85}{$t\!\coloneqq_0\! 1$}} (q_14)
 	(q_16) edge[above] node {\scalebox{0.85}{$b_0\!\coloneqq_0\! 1$}} (q_15)

 	(q_3) edge[above, bend right] node {\scalebox{0.85}{$crit_1;b_1\!\coloneqq_1\! 0$}} (q_1)
 	(q_14) edge[below, bend right] node {\scalebox{0.85}{$crit_0;b_0\!\coloneqq_0\! 0$}} (q_16)
  	;
\end{tikzpicture}
\caption{A first topological abstraction of the HDA in Figure \ref{peterHDA}}\label{peterHDA2}
\end{figure}
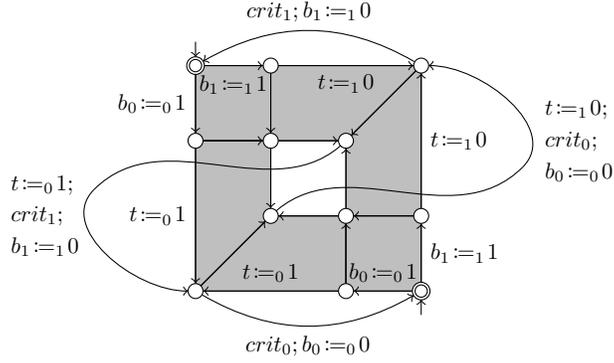

Let $\B = (P,I,F,\lambda)$ be the the HDA in Figure \ref{peterHDA2}, and let
$x$ be the upper right  $2\mbox{-}$cube of $\B$. This  $2$-cube is regular and has a free back face $d^1_ix$, which is its lower horizontal edge. Consider the precubical set $Q = P \setminus \star(d^ 1_ix)$ and the HDA $\C = (Q,I,F,\lambda|_{Q_1})$. We claim that $\C \we \B$. Since $e^1_ix = d^0_{3-i}x$ (the left vertical edge of $x$) is not the only edge ending in $d^0_1d^1_ix$, we cannot argue using Theorem \ref{freefacecollapse} and have  to use the results of Sections \ref{TC} and \ref{HG} to check the defining properties of topological abstraction individually. Homotopy equivalence is clear. Property (1) of \ref{topabdef} follows from Lemma \ref{minmax}. We are precisely in the situation of the example in Section \ref{anotherex} and have therefore an isomorphism of homology graphs. Note that this also follows from the fact that $a\to_P b$ for all $a,b \in P_0$. In order to show that we have an isomorphism of trace categories, we check the conditions of Theorem \ref{TC2}. It is clear that the first condition holds. Let $y$ denote the unique edge ending in $d^0_1d^0_1x$. Let $\omega \in Q^{\mathbb I}$ be a path from ${I \cup F \cup m_0(P) \cup m_1(P) \cup \{d^1_1d^ 1_1x\}}$ to $d^0_1d^1_ix$. Then $\omega$ passes through the upper initial state, and therefore $[\omega]$ is right-divisible in $\vec \pi_1(Q)$ by $[y_{\sharp}\cdot (e^1_{i}x)_{\sharp}]$ and hence by $[(e^1_ix)_{\sharp}]$. By Proposition \ref{dihocancel}, $y_{\sharp}\cdot (e^1_{i}x)_{\sharp}$ has the right dihomotopy cancellation property in $Q$. Since $y$ is the only edge ending in $d^0_1d^0_1x$, it follows that $(e^1_{i}x)_{\sharp}$ also has the right dihomotopy cancellation property in $Q$. Therefore the second condition of Theorem \ref{TC2} holds, and we have an isomorphism of trace categories. Thus $\C \we \B$. 

We can now use Theorem \ref{freefacecollapse} to eliminate the upper left square of $\C$ and its right vertical edge. By Theorem \ref{red2}, we can then collapse the star of the upper right vertex of the lower $2$-cube in the left column. Proceeding as before, we eliminate the remaining three $2$-cubes of $\B$. Merging edges  we then obtain our final topological abstraction of the HDA for Peterson's algorithm in Figure \ref{peterHDAmin}. It can be shown that this HDA model for Peterson's algorithm is minimal in the sense that it is a topological abstraction of $\A$ with a minimal number of cubes.

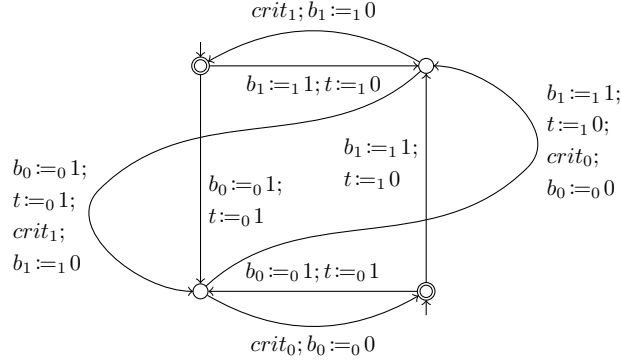
\begin{figure}[]
\begin{tikzpicture}[initial text={},on grid]

\draw[] (4,0) to [out=225,in=45] (-0.35,-1.65); 

\draw[->] (-0.35,-1.65) to [out=225,in=180] (0.9,-3);

\node[align=left] at (-1,-2) {\scalebox{0.85}{$b_0\!\coloneqq_0\! 1;$}\\\scalebox{0.85}{$t\!\coloneqq_0\! 1;$}\\\scalebox{0.85}{$crit_1;$}\\\scalebox{0.85}{$b_1\!\coloneqq_1\! 0$}};

\draw[] (1,-3) to [out=45,in=225] (5.35,-1.35); 

\draw[->] (5.35,-1.35) to [out=45,in=0] (4.1,0);

\node[align=left] at (6.1,-1) {\scalebox{0.85}{$b_1\!\coloneqq_1\! 1;$}\\\scalebox{0.85}{$t\!\coloneqq_1\! 0;$}\\\scalebox{0.85}{$crit_0;$}\\\scalebox{0.85}{$b_0\!\coloneqq_0\! 0$}};

\node[align=left] at (1.6,-1.8) {\scalebox{0.85}{$b_0\!\coloneqq_0\! 1;$}\\\scalebox{0.85}{$t\!\coloneqq_0\! 1$}}; 

\node[align=left] at (3.4,-1.3) {\scalebox{0.85}{$b_1\!\coloneqq_1\! 1;$}\\\scalebox{0.85}{$t\!\coloneqq_1\! 0$}};

\node[state,minimum size=0pt,inner sep =2pt,fill=white,accepting,initial,initial where=above,initial distance=0.2cm] (q_1) at (1,0) {};

\node[state,minimum size=0pt,inner sep =2pt,fill=white] (q_3) [right=of q_2,xshift=1cm] {};

\node[state,minimum size=0pt,inner sep =2pt,fill=white,accepting,initial,initial where=below,initial distance=0.2cm] (q_16) [left=of q_17,xshift=0cm] {};

\node[state,minimum size=0pt,inner sep =2pt,fill=white] (q_14) [left=of q_15,xshift=-1cm] {};

    \path[->] 
    
    (q_1) edge[below] node {\scalebox{0.85}{$b_1\!\coloneqq_1\! 1;t\!\coloneqq_1\! 0$}} (q_3)

   	(q_1) edge[right] node{}(q_14)

 	(q_16) edge[right] node {} (q_3)

 	(q_16) edge[above] node {\scalebox{0.85}{$b_0\!\coloneqq_0\! 1;t\!\coloneqq_0\! 1$}} (q_14)

 	(q_3) edge[above, bend right] node {\scalebox{0.85}{$crit_1;b_1\!\coloneqq_1\! 0$}} (q_1)
 	(q_14) edge[below, bend right] node {\scalebox{0.85}{$crit_0;b_0\!\coloneqq_0\! 0$}} (q_16)
  	;
\end{tikzpicture}
\caption{Final reduced HDA model for Peterson's algorithm}\label{peterHDAmin}
\end{figure}

\subsection{Properties of Peterson's algorithm} 

The most important property of Peterson's algorithm is, of course, mutual exclusion. This safety property can be expressed as the  $\A$-dihomotopy invariant property $$\bigcap _{i=0}^1\Sigma^* \setminus \left(\Sigma ^* \cdot \{crit_i\}\cdot (\Sigma\setminus \{b_i:=_i0\})^ * \cdot \{crit_{1-i}\}\cdot \Sigma^ *\right),$$
at least if we restrict our attention to successful finite executions of the system. We remark that in order to check that such a property is $\A$-dihomotopy invariant, it is not necessary to have full information on the HDA $\A$: it suffices  to know its local independence relation (see Example \ref{exproperty}). By Theorem \ref{dihoprop}, in order to verify that $\A$ has the above mutual exclusion property, it is enough to check it for the smaller HDA in Figure \ref{peterHDAmin}. 

Since $\A$ is coaccessible, the fact that it has the above mutual exclusion property implies that it has the mutual exclusion property for arbitrary finite and infinite executions. Indeed, otherwise mutual exclusion would fail to hold in some finite execution, which could be extended to a successful finite execution violating the mutual exclusion property. Note that the coaccessibility of $\A$ can be derived from the coaccessibility of the HDA in Figure \ref{peterHDAmin} using the results of the previous subsections and Proposition \ref{coacc}, which applies in the case of the elementary collapses where we could not use Theorem \ref{freefacecollapse}.

The coaccessibility of $\A$ also implies that the system given by Peterson's algorithm is deadlock-free. This also follows from the fact that the HDA in Figure \ref{peterHDAmin}---and hence $\A$---has no maximal vertices. 

Another important feature of Peterson's algorithm is starvation freedom: A process that requests access to the critical section will eventually obtain it. This liveness property can be inferred from $\A$-dihomotopy invariant properties, which may be verified for the HDA in Figure \ref{peterHDAmin} instead of for $\A$ itself. For successful finite executions, starvation freedom is the $\A$-dihomotopy invariant property
$$\bigcap \limits_{i=0}^1 \Sigma^* \setminus \left(\Sigma ^* \cdot \{b_i:=_i1\} \cdot (\Sigma \setminus \{crit_i\})^*\right).$$
For infinite executions, starvation freedom follows from the following $\A$-dihomotopy invariant properties where $a \in \Sigma \setminus \{crit_0,crit_1\}$ and $i \in \{0,1\}$:
$$\begin{array}{l} \bullet\;\; \Sigma^* \setminus \left(\Sigma^* \cdot \{a\} \cdot (\Sigma\setminus \{crit_0,crit_1\})^*\cdot \{a\}\cdot \Sigma^*\right)\\
\bullet\;\;  \Sigma^* \setminus \left(\Sigma ^* \cdot \{b_i:=_i1\} \cdot (\Sigma \setminus \{crit_i\})^*\cdot \{crit_{1-i}\} \cdot (\Sigma\setminus \{crit_i\})^*\cdot \{crit_{1-i}\}\cdot \Sigma^*\right)\end{array}$$


\begin{thebibliography}{} \label{biblio}
\bibitem{BaierKatoen} C. Baier, J.-P. Katoen, \emph{Principles of Model Checking}, The MIT Press (2008).



\bibitem{BubenikW} P. Bubenik and K. Worytkiewicz, A model category for local po-spaces, \emph{Homology, Homotopy and Applications} 8 (2006),  263-292.

\bibitem{BubenikExtremal} P. Bubenik, Models and Van Kampen theorems for directed homotopy theory, Homology, Homotopy and Applications 11 (1) (2009), 185-202.

\bibitem{ClarkeGrumbergPeled} E.M. Clarke, Jr., O. Grumberg, D.A. Peled, \emph{Model Checking}, The MIT Press (1999).

\bibitem{CridligGoubault} R. Cridlig, E. Goubault, Semantics and Analysis of Linda-based Languages, WSA'93, Lecture Notes in Computer Science 724, Springer (1993), 72-86.

\bibitem{Diekert} V. Diekert, Combinatorics on Traces, Lecture Notes in Computer Science 454, Springer-Verlag (1990).

\bibitem{FahrenbergThesis} U. Fahrenberg, Higher-dimensional automata from a topological viewpoint, Ph.D. thesis, Aal\nolinebreak borg University, Denmark (2005).



\bibitem{FajstrupGR} L. Fajstrup, M. Rau\ss en, E.  Goubault, Algebraic topology and concurrency, Theoretical Computer Science 357 (2006), 241-278.

\bibitem{Components} L. Fajstrup, M. Raussen, E. Goubault, E. Haucourt, Components of the fundamental category, Applied Categorical Structures 12 (2004), 81-108.

\bibitem{GaucherFlowModel} P. Gaucher, A model category for the homotopy theory of concurrency, Homology, Homotopy and Applications 5 (1) (2003), 549-599.

\bibitem{GaucherProcess} P. Gaucher, Towards a homotopy theory of process algebra, Homology, Homotopy and Applications 10 (1) (2008), 353-388.

\bibitem{GaucherT} P. Gaucher, T-homotopy and refinement of observation (I) : Introduction, Electronic Notes in Theoretical Computer Science 230 (2009), 103-110.



\bibitem{GaucherGoubault} P. Gaucher, E. Goubault, Topological deformation of higher-dimensional automata, Homology, Homotopy and Applications 5 (2) (2003), 39-82.

\bibitem{vanGlabbeek} R.J. van Glabbeek, On the expressiveness of higher-dimensional automata, Theoretical Computer Science 356 (2006), 265-290.

\bibitem{Godefroid} P. Godefroid, Partial-Order Methods for the Verification of Concurrent Systems - An Approach to the State-Explosion Problem, 
Lecture Notes in Computer Science
1032, Springer (1996). 

\bibitem{GoubaultDomains} E. Goubault, Domains of Higher-Dimensional Automata, CONCUR'93, 
Lecture Notes in Computer Science
715, Springer (1993), 293-307. 

\bibitem{Goubault} E. Goubault, Some geometric perspectives in concurrency theory, Homology, Homotopy and Applications 5 (2) (2003), 95-136.

\bibitem{Components2} E. Goubault, E. Haucourt, Components of the fundamental category II, Applied Categorical Structures 15 (2007), 387-414.

\bibitem{GoubaultJensen} E. Goubault, T.P. Jensen, Homology of higher-dimensional automata, in: R. Cleaveland (Ed.), Proc. CONCUR '92, Third Internat. Conf. on Concurrency Theory, Stony Brook, NY, USA, August 1992, Lecture Notes in Computer Science 630, Springer-Verlag (1992), 254-268.

\bibitem{GoubaultMimram} E. Goubault, S. Mimram, Formal relationships between geometrical and classical models for concurrency, Electronic Notes in Theoretical Computer Science 283 (2012), 77-109.

\bibitem{GrandisBook} M. Grandis, Directed Algebraic Topology - Models of Non-Reversible Worlds, New Mathematical Monographs 13, Cambridge University Press (2009).



\bibitem{reldi} T. Kahl, Relative directed homotopy theory of partially ordered spaces, \emph{Journal of Homotopy and Related Structures} 1 (1) (2006), 79-100.

\bibitem{dicubes2d} T. Kahl, Some collapsing operations for 2-dimensional precubical sets, \emph{Journal of Homotopy and Related Structures} 7(2) (2012), 281-298. 

\bibitem{weakmor} T. Kahl, Weak morphisms of higher-dimensional automata, Theoretical Computer Science 536 (2014), 42-61.

\bibitem{hgraph} T. Kahl, The homology graph of a precubical set, \emph{Homology, Homotopy and Applications}, 16 (1) (2014), 119-138. 



\bibitem{Krishnan} S. Krishnan, Cubical approximation for directed topology I, Applied Categorical Structures 23 (2)  (2015), 177-214.


\bibitem{MannaPnueli} Z. Manna, A. Pnueli, \emph{The Temporal Logic of Reactive and Concurrent Systems: Specification}, Springer-Verlag (1992).

\bibitem{Mazurkiewicz} A. Mazurkiewicz, Trace theory, in: W. Brauer, W. Reisig, G. Rozenberg (eds.), Petri Nets: Applications and Relationships to Other Models of Concurrency, Lecture Notes in Computer Science 255,  Springer-Verlag (1987), 279-324.

\bibitem{Mazurkiewicz2} A. Mazurkiewicz, Introduction to Trace Theory, in: V. Diekert, G. Rozenberg (eds.), The Book of Traces, World Scientific (1995), 3-41.

\bibitem{Misamore} M.D. Misamore, Computing path categories of finite directed cubical complexes, \emph{Applicable Algebra in Engineering, Communication and Computing} 26 (1-2) (2015), 151-164.

\bibitem{Peled} D. Peled, All
from One, One for All: on Model Checking Using Representatives, Proceedings of CAV'93, 
Lecture Notes in Computer Science 697, Springer (1993), 409-423.  

\bibitem{Peterson} G.L. Peterson, Myths about the mutual exclusion problem, Information Processing Letters 12 (3) (1981), 115-116.

\bibitem{Pratt} V. Pratt, Modeling Concurrency with Geometry, POPL '91,  Proceedings of the 18th ACM SIGPLAN-SIGACT symposium on Principles of programming languages (1991), 311-322.

\bibitem{RaussenInvariants} M. Rau\ss en, Invariants of Directed Spaces, \emph{Applied Categorical Structures} 15 (2007), 355-386.

\bibitem{Sakarovitch} J. Sakarovitch, Elements of Automata Theory, Encyclopedia of Mathematics and Its Applications, Cambridge University Press (2009).



\bibitem{WinskelNielsen} G. Winskel, M. Nielsen, Models for concurrency, \emph{Handbook of logic in computer science (vol. 4): semantic modelling}, Oxford University Press (1995), 1-148.

\end{thebibliography}
\end{document}